%% file: paper.tex
\documentclass[a4paper,UKenglish,cleveref, autoref, thm-restate]{lipics-v2021}



\EventEditors{Jonathan Aldrich and Guido Salvaneschi}
\EventNoEds{2}
\EventLongTitle{38th European Conference on Object-Oriented Programming (ECOOP 2024)}
\EventShortTitle{ECOOP 2024}
\EventAcronym{ECOOP}
\EventYear{2024}
\EventDate{September 16--20, 2024}
\EventLocation{Vienna, Austria}
\EventLogo{}
\SeriesVolume{313}
\ArticleNo{26}

\input{style}

\usepackage[ruled,vlined,linesnumbered]{algorithm2e}
\usepackage[noend]{algpseudocode}
\usepackage{amsfonts}
\usepackage{amsmath}
\usepackage{amssymb}
\usepackage[altpo,epsilon]{backnaur}
\usepackage{booktabs}
\usepackage{multicol}
\usepackage{subcaption}
\usepackage{xspace}

\usepackage[capitalise]{cleveref}
\crefname{algorithm}{Alg.}{Algs.}
\crefname{section}{Sec.}{Secs.}
\crefname{definition}{Def.}{Defs.}
\crefname{table}{Tab.}{Tabs.}
\crefname{example}{Ex.}{Exs.}
\crefname{proposition}{Prop.}{Props.}
\crefname{theorem}{Thm.}{Thms.}
\crefname{corollary}{Cor.}{Cors.}
\crefname{appendix}{Appx.}{Appxs.}

\renewcommand{\implies}{\Rightarrow}
\renewcommand{\iff}{\Leftrightarrow}

\newcommand{\auxlineref}[1]%
  {line~\ref{#1}\xspace}
\newcommand{\linerange}[2]%
  {lines~\ref{#1}--\ref{#2}\xspace}
\newcommand{\Linerange}[2]%
  {Lines~\ref{#1}--\ref{#2}\xspace}

\newcommand{\mmcode}[1]%
  {\lstinline[style=ipsmp,basicstyle=\ttfamily]{#1}}
\newcommand{\lncode}[1]%
  {\lstinline[style=ipsmp,basicstyle=\small\ttfamily]{#1}}
\newcommand{\code}[1]
  {\ifmmode\text{\mmcode{#1}}\else\lncode{#1}\fi}
\newcommand{\smallcode}[1]{{\small\code{#1}}}
\newcommand{\mmsmallcode}[1]{\text{\smallcode{#1}}}

\newcommand{\ipsmp}{IPS-MP\xspace}
\newcommand{\ipsmpfull}{Inductive Predicate Synthesis Modulo Programs\xspace}

\newcommand{\cvc}{\textsc{CVC4}\xspace}
\newcommand{\eldarica}{\textsc{Eldarica}\xspace}
\newcommand{\hornspec}{\textsc{HornSpec}\xspace}
\newcommand{\seahorn}{\textsc{SeaHorn}\xspace}
\newcommand{\smartace}{\textsc{SmartACE}\xspace}
\newcommand{\spacer}{\textsc{Spacer}\xspace}
\newcommand{\verx}{\textsc{VerX}\xspace}

\newcommand{\semgus}{\textsc{SemGuS}\xspace}
\newcommand{\sygus}{\textsc{SyGuS}\xspace}

\newboolean{supplementary}
\setboolean{supplementary}{true}
\newcommand{\appendixcite}[1]%
  {\ifthenelse{\boolean{supplementary}}%
              {\cref{#1}}%
              {the supplementary material}}
\newcommand{\proofloc}[0]%
  {\ifthenelse{\boolean{supplementary}}%
              {appendix}%
              {\text{supplementary material}}}
\newcommand{\addappendix}[1]%
  {\ifthenelse{\boolean{supplementary}}{#1}{}}

\title{Inductive Predicate Synthesis Modulo Programs (Extended)}
\author{Scott Wesley}{Dalhousie University, Halifax, Canada}{}{}{}
\author{Maria Christakis}{TU Wien, Vienna, Austria}{}{}{supported by the Vienna Science and Technology Fund (WWTF) and the City of Vienna [Grant ID: 10.47379/ICT22007].}
\author{Jorge A. Navas}{Certora, Seattle, Washington, USA}{}{}{}
\author{Richard Trefler}{University of Waterloo, Waterloo, Canada}{}{}{supported, in part, by a Discovery Grant (Individual) from the Natural Sciences and Engineering Research Council of Canada.}
\author{Valentin W{\"u}stholz}{ConsenSys, Vienna, Austria}{}{}{}
\author{Arie Gurfinkel}{University of Waterloo, Waterloo, Canada}{}{}{supported, in part, by a Discovery Grant (Individual) from the Natural Sciences and Engineering Research Council of Canada.}
\authorrunning{S. Wesley, M. Christakis, J.\,A. Navas, R. Trefler, V. W{\"u}stholz, and A. Gurfinkel}
\Copyright{Scott Wesley, Maria Christakis, Jorge A. Navas, Richard Trefler, Valentin W{\"u}stholz, and Arie Gurfinkel}

\nolinenumbers
\begin{document}
\maketitle

\begin{abstract}
  A growing trend in program analysis is to encode verification conditions within the language of the input program.
  This simplifies the design of analysis tools by utilizing off-the-shelf verifiers, but makes communication with the underlying solver more challenging.
  Essentially, the analysis tools operates at the level of input programs, whereas the solver operates at the level of problem encodings.
  To bridge this gap, the verifier must pass along proof-rules from the analysis tool to the solver.
  For example, an analysis tool for concurrent programs built on an inductive program verifier might need to declare Owicki-Gries style proof-rules for the underlying solver.
  Each such proof-rule further specifies how a program should be verified, meaning that the problem of passing proof-rules is a form of invariant synthesis.

  Similarly, many program analysis tasks reduce to the synthesis of pure, loop-free Boolean functions (i.e.,~\emph{predicates}), relative to a program.
  From this observation, we propose \ipsmpfull (\ipsmp) which extends high-level languages with minimal synthesis features to guide analysis.
  In \ipsmp, unknown predicates appear under assume and assert statements, acting as specifications modulo the program semantics.
  Existing synthesis solvers are inefficient at \ipsmp as they target more general problems.
  In this paper, we show that \ipsmp admits an \emph{efficient} solution in the Boolean case, despite being generally undecidable.  
  Moreover, we show that \ipsmp reduces to the satisfiability of constrained Horn clauses, which is less general than existing synthesis problems, yet expressive enough to encode verification tasks.
  We provide reductions from challenging verification tasks---such as parameterized model checking---to \ipsmp.
  We realize these reductions with an efficient \ipsmp-solver based on \seahorn, and describe a real-world application to smart-contract verification.
\end{abstract}

\input{introduction}
\input{overview}
\input{background}
\input{synthesis}
\input{decidability}
\input{reductions}
\input{implementation}
\input{related}
\input{conclusion}

\bibliographystyle{plainurl}
\bibliography{bibliography}

\newpage

\addappendix{\appendix}
\addappendix{\input{appendix_lang}}
\addappendix{\input{appendix_pcmc}}
\addappendix{\input{appendix_proofs}}

\end{document}

%% file: introduction.tex
\section{Introduction}
\label{Sec:Intro}

In recent years, many tools have emerged to verify C programs by leveraging the Clang/LLVM compiler infrastructure~(e.g.,~\cite{JiriLubos2010,SinzFalke2010,RienerFey2012,RakamaricEmmi2014,GurfinkelKahsai2015}).
These tools take as input C programs annotated with assumptions and assertions, and decide whether an assertion can be violated given that all assumptions are satisfied.
One such tool is \seahorn~\cite{GurfinkelKahsai2015}, which employs techniques from software model checking~\cite{KomuravelliGurfinkel2014}, abstract interpretation~\cite{GurfinkelNavas2022}, and memory analysis~\cite{KuderskiNavasGurfinkel2019} to enable efficient verification.
Due to these features, many tool designers have started using annotated C code as an intermediate language to dispatch program analysis problems to \seahorn~(e.g.,~\cite{DanOrnaSharon2018,KalraGoel2018,ElviraJesus2019,BloemJacobsVizel2019,NishanthEtAl2019,WesleyEtAl2021}).
In this setting, programs with specifications are transformed into C programs with assumptions and assertions, and then these C programs are analyzed using \seahorn.
The results obtained from \seahorn are examined to draw conclusions about the input programs.

However, the flexibility afforded by C code as an intermediate language makes communication with the underlying verification algorithm more challenging.
When \seahorn is given a program to verify, it automatically applies various builtin proof-rules, such as induction for loops~\cite{DBLP:conf/cav/AlbarghouthiLGC12} and function summarization~\cite{KomuravelliGurfinkel2014}.
A tool designer has no control over how these rules are employed, nor is the developer able to introduce new proof-rules to \seahorn.
The goal of this paper is to extend \seahorn with the language features required to communicate new declarative proof-rules to the underlying verification algorithm.

To illustrate this challenge, we consider \smartace~\cite{WesleyEtAl2021}, a tool that uses \seahorn for modular Solidity smart-contract verification.
In \smartace, each smart-contract is modeled by a non-terminating loop that executes a sequence of transactions\footnote{Transactions in Solidity/Ethereum can be thought of as sequences of method invocations.}.
For \smartace to verify a smart-contract, it first requires an inductive invariant for the non-terminating loop, and a compositional invariant for each map\footnote{In Solidity, maps are often used to store data for individual smart-contract users.} in the program.
The discovery of an inductive invariant is automated by \seahorn's invariant inference capabilities.
However, \seahorn is unaware of the modular proof-rules used by \smartace, and therefore, the end-user must provide the compositional invariants manually.
The authors of \smartace hypothesized~\cite{WesleyEtAl2021} that if each proof-rule could be declared to \seahorn, then \seahorn could instruct the underlying verification algorithm to infer all invariants automatically.
Inspired by this hypothesis, we first implemented compositional invariant synthesis in \seahorn, and then discovered that our solution generalized to many program verification problems.
Consequently, our solution forms a general-purpose framework well-suited to compositional invariant synthesis.

\begin{figure}[t]
  \centering
  \includegraphics[scale=0.6]{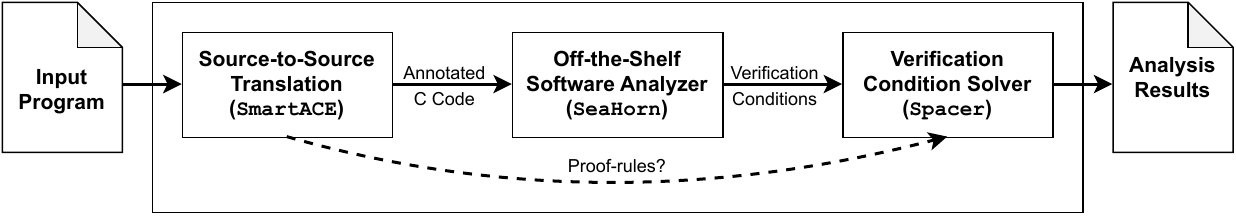}
  \caption{The architecture of an analysis framework built atop an off-the-shelf software verifier.
           Examples are given with respect to \smartace.}
  \label{Fig:Intro:ToolPipeline}
\end{figure}

To illustrate this more general problem, consider a tool designer who wishes to use an off-the-shelf software verifier (e.g.,~\seahorn) as the back-end to a new analysis framework (e.g.,~\smartace).
Recall that many off-the-shelf verifiers rely on specialized solvers to discharge verification conditions, including solvers for Satisfiability Modulo Theories~\cite{BarrettTinelli2018}, Constrained Horn Clauses (CHCs)~\cite{JaffarL87}, or intermediate verification languages (e.g.,~\cite{BarnettChang2005,FilliatrePaskevich2013}).
As depicted in \cref{Fig:Intro:ToolPipeline}, an analysis framework built atop an off-the-shelf verifier takes as input a program with specifications, translates this program into the language of the verifier, and then uses the verifier to generate verification conditions for its specialized solver.
Since software verification is undecidable in general, it is often necessary for the tool designer to declare additional proof-rules for the solver.
Example proof-rules include introducing predicate abstractions, suggesting modular abstractions for an array, and proposing modular decompositions for a parameterized system.
However, it is challenging for the tool designer to communicate proof-rules to the solver---the former operates at the level of the input program, while the latter operates at the level of verification conditions.
If a tool designer does attempt to encode proof-rules at the level of the input program, then these proof-rules are typically eliminated by optimizations from the verifier\footnote{For example, a pure function with annotations may be optimized away by the Clang compiler.}, long before verification conditions are ever produced.
That is, there is an impedance mismatch!

To bridge this gap, the verifier must pass proof-rules from the tool designer to the solver.
Each proof-rule is associated with a set of invariants that the solver must find in order to prove the program correct.
In other words, the invariants are declared by the proof-rules.
Since these invariants span many classes (e.g., inductive, compositional, and object invariants), it is often the case that specialized invariant inference techniques cannot solve this problem.
Instead, one must note that each proof-rule refines the invariants which the solver must synthesize. 
Consequently, one solution to the aforementioned impedance mismatch is to use synthesis techniques (e.g.,~\cite{AlurFisman2016,FedyukovichKaufman2017,ZhuMagill2018,SiNaik2020}).
In particular, using synthesis allows the tool designer to declare proof-rules by specifying what invariants are to be synthesized at the level of the input program.
This flexibility, however, comes at a price.
General synthesis is significantly more expensive than verification~\cite{Vardi2008}!

Our key contribution is a definition of \emph{a new form of
  synthesis}, called  \emph{Inductive Predicate Synthesis Modulo Programs}
(\ipsmp), that bridges the gap between flexible verification and efficient synthesis.
Our \emph{theoretical results} are two-fold, we show that: (a)
\ipsmp reduces to satisfiability of CHCs, hence establishing that \ipsmp is a specialization of general synthesis~\cite{Solarlezama2013,TorlakB14,AlurBodik2015,KimHu2021};
(b) for the special case of Boolean programs, \ipsmp is decidable with the same complexity as verification.
We conjecture that the latter extends to other decidable models of programs (e.g., timed automata).
Our \emph{practical result} is to reduce a wide range of common proof-rules to \ipsmp.
We show how \ipsmp guides inference of inductive invariants, class invariants, array invariants, and even modular parameterized model checking.
In other words, IPS-MP is well-suited to many areas of program analysis.
As a real-world application, we show that \ipsmp enables the full automation of \smartace.

Similar to existing synthesis frameworks, \ipsmp extends a programming language with unknowns.
The language itself is unrestricted (i.e., it has loops, procedures, memory, etc.).
However, the unknowns may only appear within \code{assume} and \code{assert} statements, denoting constraints on the strongest and weakest possible solutions, respectively.
A solution to an \ipsmp problem is a mapping from each unknown to a Boolean predicate such that the resulting program is correct (i.e., satisfies all of its assertions).
A high-level overview of \ipsmp is shown in \cref{Fig:Overview:Pipeline}.
Each problem instance consists of two components: (1)~a program with its specification (described by assumptions and assertions), which contains calls to unknown predicates under assume and assert, and (2)~the declarations of those predicates, which we refer to as \emph{predicate templates}.
Intuitively, a predicate template is a partial implementation with a number of unknown statements.
A solution consists of a full implementation of each predicate, or a witness to unrealizability (i.e.,~a proof that a solution does not exist).

\begin{figure}[t]
  \centering
  \includegraphics[scale=0.6]{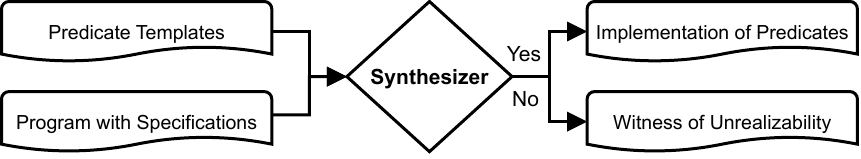}
  \caption{Overview of the \ipsmp problem.}
  \label{Fig:Overview:Pipeline}
\end{figure}

The reducibility of \ipsmp to CHC-solving motivates an efficient \ipsmp solver.
We build on the \seahorn framework (thus, our underlying language is the fragment of C supported by \seahorn~\cite{GurfinkelKahsai2015}), and integrate with two CHC solvers, namely \spacer~\cite{KomuravelliGurfinkel2014}, and \eldarica~\cite{HojjatRummer2018}.
Our empirical results on verification problems from various domains show that: (1) \ipsmp is effective at specifying verification strategies, (2) our implementation combined with existing CHC-solvers is highly efficient for linear arithmetic invariants, and (3) existing reductions to either general synthesis or specification inference are infeasible.
Our evaluation focuses on general synthesis, rather than invariant inference, since the invariants in our benchmarks span many classes.
To contextualize these results, we briefly review the state-of-the-art in synthesis.

\noindent\textbf{State-of-the-art in synthesis}.
The general synthesis problem is the task of generating a program
that satisfies a given specification.
There are many general synthesis frameworks, e.g., Sketch~\cite{Solarlezama2013}, Rosette~\cite{TorlakB14}, \sygus~\cite{AlurBodik2015}, and \semgus~\cite{KimHu2021}.
In Sketch and Rosette, users write programs with \emph{holes}, representing unknowns.
These holes are filled with predefined, loop-free expressions such that all program assertions are satisfied.
\sygus introduced a more language-agnostic approach to general synthesis. It generates loop-free \emph{programs}, satisfying a given behavioral specification, from a potentially infinite language.
Building on \sygus, \semgus allows users to define pluggable semantics, thereby enabling synthesis of programs with loops.
A distinguishing characteristic along this line of work is an emphasis on software development.
In contrast, IPS-MP targets software verification and proof synthesis, which are theoretically simpler problems.

Specification synthesis~(e.g.,~\cite{DasLahiri2015,AlbarghouthiDillig2016,PrabhuFedyukovich2021}) is another line of work that addresses a more specialized synthesis problem targeting program analysis, rather than software development.
In specification synthesis, a program may call functions with unknown implementations.
The goal is to synthesize specifications (e.g., the weakest specification for an unknown library procedure) that ensure the correctness of the calling program.
Typically, a specification synthesizer imposes extra requirements, such as non-vacuity~\cite{PrabhuFedyukovich2021}, maximality~\cite{AlbarghouthiDillig2016}, or reachability~\cite{DasLahiri2015}, to ensure that solutions are reasonable.
In contrast, the invariants synthesized by IPS-MP have constraints on both the strongest and weakest possible solutions, avoiding the need for additional (and often costly) requirements.

Of particular interest are the similarities and differences between \ipsmp and syntax-guided synthesis.
In \ipsmp, program holes are filled by expressions from an unbounded language.
To make this problem tractable, \ipsmp restricts Sketch and Rosette by requiring that holes only appear in partial predicates.
Formally, this means that \ipsmp solving is subsumed by non-linear constrained Horn clause solving.
This restriction is crucial as it allows an \ipsmp solver to prove that a problem is unrealizable, unlike in Sketch or Rosette.
Furthermore, \ipsmp differs from \sygus and \semgus in that the behavioral specification is given with respect to a given program (in other words, modulo a given program), rather than through a separate logical specification.
The program itself also places requirements on the holes, through assumptions and assertions, which is in contrast to specification synthesis.

In recent years, new extensions have been proposed to the Sketch framework.
However, these extensions all generalize Sketch to more complex, and consequently less tractable, problems, whereas \ipsmp restricts Sketch to a more tractable problem which proves to be useful in the domain of program verification.
To illustrate these gaps, we compare \ipsmp to PSKETCH~\cite{SolarLezama2008}, Synapse~\cite{Bornholt2016}, Grisette~\cite{Sirui2023}, and MetaLift~\cite{Bhatia2023}.
In the case of PSKETCH, both frameworks target the development of provably correct concurrent programs.
However, PSKETCH focuses on inductive program verification in the presence of interleaving executions, whereas \ipsmp focuses on the verification of sequential code fragments via user-defined proof rules (e.g., the synthesis of compositional invariants in \smartace).
In the case of Synapse, both tools aim to extend program synthesis problems with hints provided by an end-user.
However, the nature of these hints is very different.
In \ipsmp, the user introduces entirely new proof-rules, for which an underlying solver oversees the search for a solution.
In contrast, the hints provided by an end-user to Synapse assign costs to solutions, for which the underlying solver tries to optimize.
These hints do not allow the end-user to propose new proof-rules, and are suited to synthesis optimization rather than program verification.
In the case of Grisette, a framework was proposed to programmatically generate and solve sketches.
However, Grisette is based around bounded model-checking, whereas the \ipsmp problem targets unbounded model-checking and is, therefore, incomparable.
More closely related is MetaLift, which makes use of the fact that inductive program verification can be reduced to syntax-guided synthesis.
However, this verification program is not exposed to end-users.
In particular, the assume and assert statements are hidden from end-users, and the end-user has no way to propose new placements for them.
We conclude that \ipsmp is a novel synthesis problem.

\noindent\textbf{Constrained Horn clauses}.
A prominent approach to verification is reduction to the satisfiability of CHCs, otherwise known as \emph{verifier synthesis}~\cite{GrebenshchikovLPR12}.
While verifier synthesis does enable the flexible design of software verifiers, it does not address the issue of communicating proof-rules to the underlying solver.
In invariant synthesis, the proof-rules are either chosen once and for all~\cite{GrebenshchikovLPR12}, or are implicit in the solving algorithm (e.g.,~\cite{DBLP:journals/acta/LerouxRS16,DBLP:conf/cav/GovindCSG20}).
While we show that IPS-MP reduces to CHC-solving, our focus is on communicating new proof-rules to the solver via synthesis.
Other solutions to \ipsmp might emerge in the future.

\noindent\textbf{Contributions}.
This paper makes the following contributions:
\begin{enumerate}
\item \cref{Sect:Problem} presents the novel \ipsmp problem which has many applications to both program analysis and software verification;
\item \cref{Sect:Decidability} shows that even though \ipsmp is undecidable in general, there exists an efficient solution modulo Boolean programs;
\item \cref{Sect:Reductions} provides reductions from important verification problems to \ipsmp;
\item \cref{Sect:Implementation} presents a solver for \ipsmp within \seahorn.
      We demonstrate the effectiveness of our implementation compared to state-of-the-art synthesis frameworks \cvc~\cite{BarretConway2011} (a \sygus synthesizer) and \hornspec~\cite{PrabhuFedyukovich2021} (a specification synthesizer).
      We conclude that IPS-MP fills a gap not met by other synthesis frameworks. 
\end{enumerate}
All omitted proofs are found in the \proofloc.

%% file: overview.tex
\section{Overview}
\label{Sect:Overview}

To illustrate the basics of \ipsmp, we start with an artificial example.
For the moment, we focus on the language used in our presentation and the possible solutions to an \ipsmp problem.
Realistic applications of \ipsmp, highlighting its importance, are presented later in this section.

\begin{figure}[t]
\begin{lstlisting}[style=ipsmp,multicols=2]
°bool PRED_TEMPLATE Post(int x, int y) {°¤\label{Line:Overview:Sample:Template}¤
  °return synth(x, y); }°¤\label{Line:Overview:Sample:Pfn}¤
void main(int y) {
  int x = 0;
  assume(y > 0);¤\label{Line:Overview:Sample:AssumeY}¤
  for (int i = 0; i < y; ++i) {
    x+=1; }¤\label{Line:Overview:Sample:AddToX}¤
  ~assert(Post(x, y));~¤\label{Line:Overview:Sample:PfnAssert}¤
  ^x = *; y = *; assume(Post(x, y));^¤\label{Line:Overview:Sample:PfnAssumeMid}¤
  assert(x == y);¤\label{Line:Overview:Sample:PfnAssumeEnd}¤ }
\end{lstlisting}
  \caption{A simple example of the \ipsmp problem.}
  \label{Fig:Overview:Sample}
\end{figure}

Our example, shown in \cref{Fig:Overview:Sample}, consists of a single function \code{main} that provides the context for a synthesis problem.
The function \code{main} is written in a typical imperative language, with loops and function calls.
We extend the language with two verification statements, \code{assume} and \code{assert}, with their usual semantics.
In our example, \code{y} is initially positive, due to \code{assume(y > 0)} on \auxlineref{Line:Overview:Sample:AssumeY}, and the program is correct if \code{assert(x == y)} on \auxlineref{Line:Overview:Sample:PfnAssumeEnd} holds for all executions.
That is, lines~\ref{Line:Overview:Sample:AssumeY} and \ref{Line:Overview:Sample:PfnAssumeEnd} provide a program specification.
The goal of this example is to synthesize a pure expression $e$ such that the program is correct after substituting $e$ for each call to \code{Post}.
To indicate that a predicate is a target for synthesis, the language is extended by the predicate template annotation \code{PRED_TEMPLATE} (\auxlineref{Line:Overview:Sample:Template}).
Each predicate template is a pure, loop-free function whose body either returns \code{true}, or returns via a call to the special predicate \code{synth}.
Each call to \code{synth} indicates a \emph{hole} in the predicate implementation, and must be determined by a synthesizer.
Each return of \code{true} places an \emph{explicit constraint} on when the implementation must be true.
In our example, \code{Post} always returns via a call to \code{synth} (\auxlineref{Line:Overview:Sample:Pfn}).
In the rest of the program, a call to a partial predicate can only appear as an argument to either \code{assume} or \code{assert}.
As described below, verification calls place \emph{implicit constraints} on when the implementation must be true.
Multiple calls to the same predicate are allowed.
In our example, \code{Post} is called once under \code{assert}~(\auxlineref{Line:Overview:Sample:PfnAssert}), and once under \code{assume}~(\auxlineref{Line:Overview:Sample:PfnAssumeMid}).

\newcommand{\synthfn}{\mathit{synth}}
A \emph{solution} to an \ipsmp problem is a mapping from each partial predicate $p$ to a
pure Boolean expression $e$ over the arguments of $p$, such that if every call to
\code{synth} in $p$ is replaced by $e$, then the main program satisfies all of its
assertions.
If such a solution does not exist, the output is a \emph{witness to unrealizability}, which is a mapping from each
partial predicate $p$ to a pure Boolean expression $e$ over the arguments of $p$,
which is both necessitated by the assertions placed on the partial predicate, and
sufficient to violate an assertion that is part of the specification.
In our example, there are many possible solutions. The weakest and strongest solutions are $\mathit{post}_{weak}(x, y) = (x = y)$ and $\mathit{post}_{strong}(x, y) = (y > 0  \land x = y)$.
Each solution defines a corresponding predicate \code{Post} such that all assertions in the \code{main} program are satisfied.
Intuitively, each call to \code{Post} under \code{assume} provides an implicit constraint on the weakest possible synthesized solution.
Likewise, each call to \code{Post} under \code{assert} provides an implicit constraint on the strongest possible synthesized solution.
Following this intuition, the example shows an application of \ipsmp to find an intermediate post-condition, over two variables \code{x} and \code{y}, that is true after the loop and is strong enough to ensure an assertion.
This means that solving \ipsmp requires, in general, inferring inductive invariants for loops and summaries for functions.

To illustrate the case when synthesis is not possible, consider removing
\auxlineref{Line:Overview:Sample:AddToX} from \cref{Fig:Overview:Sample}. Since \code{x}
is not incremented, it will never equal \texttt{y}. However, \code{Post} cannot be mapped to \code{false}, since this violates the assertion on
\auxlineref{Line:Overview:Sample:PfnAssert}. If \code{Post} is not \code{false}, then the assertion on
\auxlineref{Line:Overview:Sample:PfnAssert}
is reachable and will fail. Therefore, this
\ipsmp problem is unrealizable.
The witness to unrealizability is a mapping that sends \code{Post} to an expression over \code{x} and \code{y}, which is necessitated by the assertion on \auxlineref{Line:Overview:Sample:PfnAssert} and violates the assertion on \auxlineref{Line:Overview:Sample:PfnAssumeEnd}.
An example witness is $\synthfn_{witness}( x, y ) = ( x = 0 \land y = 1 )$.

This section continues with three important applications of \ipsmp.
\cref{Sect:Overview:Methodology} presents a methodology to reduce verification problems to \ipsmp.
For readers new to verification as synthesis, the standard example of inductive loop invariant inference can be found in \appendixcite{Appendix:Loop}.
Secs.~\ref{Sect:Overview:Class}, \ref{Sect:Overview:Array}, and \ref{Sect:Overview:PCMC} extend on the techniques in \appendixcite{Appendix:Loop} to unify class invariant inference, array verification, and parameterized compositional model checking under a single synthesis framework.
\cref{Sect:Overview:Discussion} discusses the benefits of predicate templates and explains why \ipsmp requires both assumptions and assertions of partial predicates.
We note that the automation in \smartace is a special case of \cref{Sect:Overview:Array}.

\subsection{Methodology}
\label{Sect:Overview:Methodology}

In \cref{Fig:Overview:Sample}, a single predicate (i.e.,~\code{Post}) represents a single unknown (i.e.,~the post-condition of a loop).
This permits an \ipsmp solver to explore all relations between arguments (e.g.~\code{x} and \code{y} of \code{Post}).
When there are many variables, or large variable domains, the space of candidate solutions becomes very large.
Restricting the syntactic structure of each unknown, referred to as its \emph{shape}, helps to prune the search space.
In general, an unknown can be split into cases (see~\cref{Sect:Overview:Array}), and the variables in each case can be partitioned (see~\appendixcite{Appendix:Loop}).
Each partition is encoded by a unique predicate.
Refining a predicate's shape prunes the candidate solution space, but may eliminate valid solutions.

Whenever an unknown is refined, the syntactic changes are reflected only where the unknown is assumed or asserted.
The program remains unchanged otherwise.
For this reason, in \ipsmp, it is convenient to separate unknowns from their shapes.
In the context of program verification, this is accomplished with the following methodology.
First, a proof-rule for the program of interest is reduced to assumptions and assertions on one or more unknowns.
This is done once per proof-rule.
Second, the shape of each unknown is refined using insight from the program.
Third, the program is instrumented with assumptions and assertions.
The instrumented program is an \ipsmp problem and is automatically solved by an \ipsmp solver. 
We illustrate this methodology using examples from object-oriented program analysis, array verification, and parameterized verification.

\subsection{Class Invariant Inference as Synthesis}
\label{Sect:Overview:Class}

\begin{figure}[t]
  \begin{subfigure}{0.41\textwidth}
\begin{lstlisting}[style=ipsmp]
class Counter {
  int max; int pos;
  Counter(int _max) {¤\label{Line:Overview:ClassProb:InitStart}¤
    assume(_max > 0);
    max = _max;
    pos = 0; }¤\label{Line:Overview:ClassProb:InitEnd}¤
  void reset() { pos = 0; }
  int capacity() {
    return max - pos; }
  bool increment() {
    if (pos >= max) return false;
    pos += 1;
    return true; }}

void drain(Counter o) {
  while (o.capacity() > 0) {¤\label{Line:Overview:ClassProb:Loop}¤
    assert(o.increment()); }¤\label{Line:Overview:ClassProb:IncrAssert}¤
  o.reset();
  assert(o.capacity() > 0); }¤\label{Line:Overview:ClassProb:CapAssert}¤
\end{lstlisting}
    \caption{The original program.}
    \label{Fig:Overview:ClassProb}
  \end{subfigure}
  \hfill
  \begin{subfigure}{0.56\textwidth}
\begin{lstlisting}[style=ipsmp]
°bool PRED_TEMPLATE CInv(int m, int p) {°¤\label{Line:Overview:ClassSoln:Pfn}¤
  °return synth(m, p); }°
void main(void) {
  if (*) { ¤\label{Line:Overview:ClassSoln:Start1}¤
    Counter o(*);
    ~assert(CInv(o.max, o.pos));~ ¤\label{Line:Overview:ClassSoln:End1}¤
  } else if (*) { ¤\label{Line:Overview:ClassSoln:Start2}¤
    ^Counter o = *; assume(CInv(o.max, o.pos));^
    o.reset();
    ~assert(CInv(o.max, o.pos));~ ¤\label{Line:Overview:ClassSoln:End2}¤
  } else if (*) { ¤\label{Line:Overview:ClassSoln:Start3}¤
    ^Counter o = *; assume(CInv(o.max, o.pos));^
    o.increment();
    ~assert(CInv(o.max, o.pos));~ ¤\label{Line:Overview:ClassSoln:End3}¤
  } else { ¤\label{Line:Overview:ClassSoln:Start4}¤
    ^Counter o = *; assume(CInv(o.max, o.pos));^
    drain(o); }} ¤\label{Line:Overview:ClassSoln:End4}¤
\end{lstlisting}
    \caption{The \ipsmp problem (using \cref{Fig:Overview:ClassProb}).}
    \label{Fig:Overview:ClassSoln}
  \end{subfigure}
  \caption{A program (see~\cref{Fig:Overview:ClassProb}) which is correct relative to the choice of class invariant $( 0 <
    \textsf{o.max} ) \land ( 0 \le \textsf{o.pos} \le \textsf{o.max} )$, and
    a corresponding \ipsmp instance.}
\end{figure}

As a first example, we illustrate a reduction from class invariant inference to \ipsmp.
In object-oriented programming, a class bundles together a data structure, its initialization procedure, and its operations.
For example, the \code{Counter} class in \cref{Fig:Overview:ClassProb} accumulates values between $0$ and some maximum value.
The underlying data structure is a pair consisting of the current value, \code{pos}, and the maximum value, \code{max}.
The initialization procedure on \linerange{Line:Overview:ClassProb:InitStart}{Line:Overview:ClassProb:InitEnd} first ensures that \code{\_max} is positive, and then sets the current value to $0$ and the maximum value to \code{\_max}.
The operations for \code{Counter} include \code{reset}, \code{capacity}, and \code{increment}.
When \code{reset} is called, the current value is set back to $0$.
When \code{capacity} is called, the distance to the maximum value is returned.
When \code{increment} is called, if \code{capacity} is greater than $0$, then the current value is incremented and \code{true} is returned, else the current value is unchanged and \code{false} is returned.

The goal of this example is to prove that \code{drain} satisfies its assertions.
The \code{drain} function takes an instance of \code{Counter} (in an arbitrary state), exhausts the counter's capacity, and then resets the counter to 0.
The function is correct if \code{increment} always returns \code{true} on \auxlineref{Line:Overview:ClassProb:IncrAssert}, and \code{capacity} always returns a positive value on \auxlineref{Line:Overview:ClassProb:CapAssert}.
Verifying these claims is non-trivial, as the correctness of \code{drain} depends on the possible states of \code{Counter}.
For example, proving the assertion on \auxlineref{Line:Overview:ClassProb:IncrAssert} requires the invariant $( 0 \le \code{max} - \code{pos} )$.

A common approach to the modular analysis of object-oriented programs is class invariant inference (e.g.,~\cite{Ernst2000,HuizingKuiper2000,AggarwalRandall2001,Logozzo2004}).
A class invariant is a predicate that is true of a class instance after initialization, closed under the application of impure class methods, and sufficient to prove the correctness of the class~\cite{HuizingKuiper2000}.
In the case of \code{Counter}, the impure methods are \code{reset} and \code{increment}.
Therefore, a class invariant for \code{Counter} must satisfy four requirements.

\cref{Fig:Overview:ClassSoln} illustrates a technique to encode multiple cases in a single \ipsmp program.
Intuitively, this program uses non-determinism to execute one of four possible cases.
A case is selected on \auxlineref{Line:Overview:ClassSoln:Start1} by a sequence of \code{if}-\code{else} statements, each with a non-deterministic condition.
Even though the execution of each case is mutually exclusive, the \ipsmp solution must work in all cases.
The cases in \cref{Fig:Overview:ClassSoln} correspond to the requirements of a class invariant for \code{Counter}.
To ensure that the class invariant holds after initialization, the first case initializes an instance of \code{Counter} with non-deterministic arguments, and then asserts that the instance satisfies the class invariant (\linerange{Line:Overview:ClassSoln:Start1}{Line:Overview:ClassSoln:End1}).
To ensure that the class invariant is closed with respect to \code{reset}, the second case selects an arbitrary instance of \code{Counter} (through non-determinism), assumes that this instance satisfies the class invariant, executes \code{reset}, and then asserts that the instance continues to satisfy the class invariant (\linerange{Line:Overview:ClassSoln:Start2}{Line:Overview:ClassSoln:End2}).
Similarly, the third case ensures that the class invariant is closed with respect to \code{increment} (\linerange{Line:Overview:ClassSoln:Start3}{Line:Overview:ClassSoln:End3}).
Finally, to ensure that the class invariant entails the correctness of \code{drain}, the fourth case selects an arbitrary instance of \code{Counter}, assumes that this instance satisfies the class invariant, and then calls \code{drain} with the instance as an argument (\linerange{Line:Overview:ClassSoln:Start4}{Line:Overview:ClassSoln:End4}).
This gives a program with unknowns, as required by the verification methodology.

Next, the shape of the class invariant is considered.
In this example, we lack program-specific knowledge to help split the invariant into sub-cases.
Furthermore, it would be futile to partition the invariant's arguments, as the invariant must relate \code{max} to \code{pos} (e.g., \auxlineref{Line:Overview:ClassProb:IncrAssert} of \cref{Fig:Overview:ClassProb} requires that $0 \le \code{max} - \code{pos}$).
Therefore, $\code{CInv}( m, p )$ is used as the shape of the invariant.
In \cref{Fig:Overview:ClassSoln}, \code{CInv} corresponds to the partial predicate on \auxlineref{Line:Overview:ClassSoln:Pfn}.
One solution to \cref{Fig:Overview:ClassSoln} is the expression \mbox{\code{(m > 0) \&\& (p <= m)}} for the hole in \code{CInv}.
To prove the correctness of \code{drain}, a synthesizer may also infer the invariant $( 0 \le \code{o.pos} \land \code{o.pos} \le \code{o.max} )$ for the loop on \auxlineref{Line:Overview:ClassProb:Loop} of \cref{Fig:Overview:ClassProb}.

\subsection{Verification of Array-Manipulating Programs as Synthesis}
\label{Sect:Overview:Array}

\begin{figure}[t]
  \begin{subfigure}{0.41\textwidth}
\begin{lstlisting}[style=ipsmp]
void main(int sz) {
  assume(sz > 0); ¤\label{Line:Overview:ArrayProb:SetupStart}¤
  int *data = new int[sz];
  memset(data, 0, sz * sizeof(int));¤\label{Line:Overview:ArrayProb:SetupEnd}¤
  int max = 0; int sid = *; ¤\label{Line:Overview:ArrayProb:SetGlobals}¤
  assume(0 <= sid && sid < sz);¤\label{Line:Overview:ArrayProb:AssumeGlobals}¤
  while (true) { ¤\label{Line:Overview:ArrayProb:LoopStart}¤
    int id = *; ¤\label{Line:Overview:ArrayProb:SelectId}¤
    assume(0 <= id && id < sz);¤\label{Line:Overview:ArrayProb:RestrictIdEnd}¤
    int v = data[id];
    if (id != sid) {v += 1;} ¤\label{Line:Overview:ArrayProb:WriteBranch}¤
    if (v > max) { max = v; } ¤\label{Line:Overview:ArrayProb:MaxBranch}¤
    data[id] = v; }} ¤\label{Line:Overview:ArrayProb:LoopEnd}¤
\end{lstlisting}
    \caption{The original program.}
    \label{Fig:Overview:ArrayProb}
  \end{subfigure}
  \hfill
  \begin{subfigure}{0.56\textwidth}
\begin{lstlisting}[style=ipsmp]
°bool PRED_TEMPLATE Inv3(int m, int v) {°¤\label{Line:Overview:ArraySoln:Pfn1}¤
  °if (m == 0 && v == 0) { return true; }°¤\label{Line:Overview:ArraySoln:BaseCase1}¤
  °else { return synth(m, v); } }°
°bool PRED_TEMPLATE Inv4(int m, int v){°¤\label{Line:Overview:ArraySoln:Pfn2}¤
  °if (m == 0 && v == 0) { return true; }°¤\label{Line:Overview:ArraySoln:BaseCase2}¤
  °else { return synth(m, v); }}° ¤\label{Line:Overview:ArraySoln:Pfn2End}¤
void main(int sid) { ¤\label{Line:Overview:ArraySoln:MainIn}¤
  int max = 0;
  while (true) { ¤\label{Line:Overview:ArraySoln:Loop}¤
    int id = *;¤\label{Line:Overview:ArraySoln:SelectId}¤
    int x = *;¤\label{Line:Overview:ArrayProb:SelectOtr}¤
    assume(id != x);¤\label{Line:Overview:ArrayProb:RestrictOtr}¤
    // int v = data[id];
    ^int v = *; int u = *;^ ¤\label{Line:Overview:ArraySoln:AssumeStart}¤
    ^if (id == sid) { assume(Inv3(max,v)); }^¤\label{Line:Overview:ArraySoln:Inv3Called}¤
    ^else { assume(Inv4(max,v)); }^ ¤\label{Line:Overview:ArraySoln:Inv4Called}¤
    ^if (x == sid){ assume(Inv3(max,u)); }^
    ^else { assume(Inv4(max,u)); }^ ¤\label{Line:Overview:ArraySoln:AssumeEnd}¤
    // Properties.
    assert(v <= max); ¤\label{Line:Overview:ArraySoln:PropStart}¤
    if (id == sid) { assert(v == 0); } ¤\label{Line:Overview:ArraySoln:PropEnd}¤
    // Update.
    if (id != sid) { v += 1; } ¤\label{Line:Overview:ArrayProb:UpdateStart}¤
    if (v > max) { max = v; } ¤\label{Line:Overview:ArrayProb:UpdateEnd}¤
    // data[id] = v;
    ~if (id == sid) { assert(Inv3(max,v)); }~¤\label{Line:Overview:ArraySoln:AssertStart}¤
    ~else { assert(Inv4(max,v)); }~
    ~if (x == sid) { assert(Inv3(max,u)); }~
    ~else { assert(Inv4(max,u)); }~}} ¤\label{Line:Overview:ArraySoln:AssertEnd}¤¤\label{Line:Overview:ArraySoln:MainOut}¤
\end{lstlisting}
    \caption{The \ipsmp problem.}
    \label{Fig:Overview:ArraySoln}
  \end{subfigure}
  \caption{A program (see~\cref{Fig:Overview:ArrayProb}) which is correct relative to the choice of array abstraction $( i =
    \textsf{s} \land v = 0 ) \lor ( i \ne \textsf{s} \land 0 \le v \le
    \textsf{max} )$, and a corresponding \ipsmp instance.}
\end{figure}

Consider the array-manipulating program in \cref{Fig:Overview:ArrayProb}.
This program initializes the array \code{data}, and then performs an unbounded sequence of updates to the cells of \code{data} while maintaining the maximum element of \code{data} in \code{max}.
A special index, stored by \code{sid}, remains unchanged during execution.
On \linerange{Line:Overview:ArrayProb:SetupStart}{Line:Overview:ArrayProb:SetupEnd}, \code{data} is allocated and then zero-initialized.
On \auxlineref{Line:Overview:ArrayProb:SetGlobals}, \code{max} is set to $0$, since the maximum element of a zero-initialized array is~$0$.
On \auxlineref{Line:Overview:ArrayProb:AssumeGlobals}, \code{sid} is set to an arbitrary index in \code{data}.
The unbounded sequence of updates begins on \auxlineref{Line:Overview:ArrayProb:LoopStart}, when the program enters a non-terminating loop.
During each iteration, an index is selected (\linerange{Line:Overview:ArrayProb:SelectId}{Line:Overview:ArrayProb:RestrictIdEnd}), and if this index is not \code{sid}, then the corresponding cell in \code{data} is incremented by $1$ (\auxlineref{Line:Overview:ArrayProb:WriteBranch}).
If the cell is incremented, then \code{max} is updated accordingly (\auxlineref{Line:Overview:ArrayProb:MaxBranch}).
Note that \cref{Fig:Overview:ArrayProb} can be thought of as a simplified smart contract, where \code{data} is an \emph{address mapping}, \code{sid} is an \emph{address variable}, and each iteration of the loop is a \emph{transaction }.
For a more general presentation of smart contracts as array-manipulating programs, see \cite{WesleyEtAl2021}.

The goal of this example is to prove two properties about the cells of \code{data}.
The first property is that every cell of \code{data} is at most \code{max}.
The second property is that \code{data[sid]} is always zero.
It is not hard to see why these properties are true.
For example, the first property is true since every cell of \code{data} is initially zero, and after increasing the value of a cell, \code{max} is updated accordingly.
However, verifying these properties is challenging, since \code{data} has an arbitrary number of cells.
One solution to this problem is to compute a summary for each cell of \code{data}, with respect to \code{max} and \code{sid}, and independent of \code{data}'s length.
This summary is then used in place of each array access to obtain a new program with a finite number of cells.
For simplicity, we assume that array accesses are in bounds, and that integers cannot overflow (i.e., are modeled as mathematical integers). 

A common approach to obtain such a summary is to
over-approximate the least fixed point of the program by an
\emph{abstract domain} that provides a tractable set of array cell
partitions according to semantic properties
(e.g.,~\cite{GopanReps2005,HalbwachsPeron2008,CousotLogozzo2011}).
An alternative approach (followed here) is to compute a \emph{compositional invariant}~\cite{NamjoshiTrefler2016} for each cell of the array.
A compositional (array) invariant is a predicate that is initially true of all cells in the array, and closed under every write to the array.
Furthermore, a compositional invariant must be closed under \emph{interference}, that is, if $i \ne j$ and the cell \code{data[i]} is updated, then \code{data[j]} continues to satisfy the compositional invariant.
That is, a compositional invariant is assumed after each read and asserted after each write.

Using this approach, the program in \cref{Fig:Overview:ArraySoln} is obtained. 
On \auxlineref{Line:Overview:ArraySoln:SelectId}, an arbitrary index named \code{id} is selected, as in the original program.
However, on \linerange{Line:Overview:ArrayProb:SelectOtr}{Line:Overview:ArrayProb:RestrictOtr}, a second, \emph{distinct} index named \code{x} is selected, to stand for a cell under interference.
On \linerange{Line:Overview:ArraySoln:AssumeStart}{Line:Overview:ArraySoln:AssumeEnd}, the compositional invariant is assumed, in place of reading the values at \mbox{\code{data[id]}} and \code{data[x]}.
On \linerange{Line:Overview:ArraySoln:PropStart}{Line:Overview:ArraySoln:PropEnd}, the two properties are asserted.
If an arbitrary cell satisfies both properties, then every cell must satisfy both properties.
On \linerange{Line:Overview:ArrayProb:UpdateStart}{Line:Overview:ArrayProb:UpdateEnd}, the cell updates are performed as in the original program.
On \linerange{Line:Overview:ArraySoln:AssertStart}{Line:Overview:ArraySoln:AssertEnd}, the compositional invariant is asserted, in place of writing to \code{data[id]}.
Note that \linerange{Line:Overview:ArrayProb:SetupStart}{Line:Overview:ArrayProb:SetupEnd} of \cref{Fig:Overview:ArrayProb} do not appear in \cref{Fig:Overview:ArraySoln} since the compositional invariant abstracts away the contents of \code{data}.
This gives a program with unknowns, as required by the verification methodology.

Next, the shape of the compositional invariant is restricted.
Observe that on \auxlineref{Line:Overview:ArrayProb:WriteBranch} of \cref{Fig:Overview:ArrayProb}, the value written into \code{data} depends on whether the index is \code{sid}.
This suggests that the compositional invariant has two cases that branch on whether \code{id} equals \code{sid}, namely $( ( \code{id} = \code{sid} ) \land \code{Inv3}( \code{max}, v ) ) \lor ( ( \code{id} \ne \code{sid} ) \land \code{Inv4}( \code{max}, v ) )$.
In the \ipsmp encoding, both \code{Inv3} and \code{Inv4} correspond to partial predicates (see lines~\ref{Line:Overview:ArraySoln:Pfn1} and \ref{Line:Overview:ArraySoln:Pfn2} in \cref{Fig:Overview:ArraySoln}, respectively).
The templates, on lines~\ref{Line:Overview:ArraySoln:BaseCase1} and \ref{Line:Overview:ArraySoln:BaseCase2}, correspond to the initialization rule for the invariant.
Recall, however, that these templates are not strictly necessary.
One alternative is to assert \code{Inv3(max,0)} and \code{Inv4(max,0)} before \auxlineref{Line:Overview:ArraySoln:Loop}, though this is not illustrated.
Due to the branching shape of the invariant, each \code{assume} and \code{assert} statement must branch between the two partial predicates (see \linerange{Line:Overview:ArraySoln:AssumeStart}{Line:Overview:ArraySoln:AssumeEnd} and \ref{Line:Overview:ArraySoln:AssertStart}--\ref{Line:Overview:ArraySoln:AssertEnd}).
Given \cref{Fig:Overview:ArraySoln}, a synthesizer may find the expressions \code{(v == 0)} for the hole in \code{Inv3}, and \code{(0 <= v) \&\& (v <= max)} for the hole in \code{Inv4}.
By substitution, $( ( \code{id} = \code{sid} ) \land ( v  = 0 ) ) \lor ( ( \code{id} \ne \code{sid} ) \land ( 0 \le v ) \land ( v \le \code{max} ) )$.
To verify \code{main}, a synthesizer may also infer the invariant $( 0 \le \code{max} )$ for the loop at \auxlineref{Line:Overview:ArraySoln:Loop}.

\subsection{Parameterized Verification as Synthesis}
\label{Sect:Overview:PCMC}

\begin{figure}[t]
  \begin{subfigure}{0.48\textwidth}
\begin{lstlisting}[style=ipsmp]
typedef enum { Free, Left, Right } Lock;
typedef enum { Try, Critical } State;
struct View { ¤\label{Line:Overview:PcmcProb:View}¤
  Lock lhs; Lock rhs; State st; };
View tr(View v) { ¤\label{Line:Overview:PcmcProb:Tr}¤
  bool held = v.lhs == Left &&
              v.rhs == Right
  if (v.st == Critical) {
    v.st = Try;
    v.lhs = Free;
    v.rhs = Free; }
  else if (held) {
    v.st = Critical; }
  else if (v.lhs == Free) {
    v.lhs = Left; }
  else if (v.rhs == Free) {
    v.rhs = Right; }
  return v; }
\end{lstlisting}
    \caption{The process.}
    \label{Fig:Overview:PcmcProb}
  \end{subfigure}
  \hfill
  \begin{subfigure}{0.51\textwidth}
\begin{lstlisting}[style=ipsmp]
°bool PRED_TEMPLATE RInv(° ¤\label{Line:Overview:PcmcSoln:Pfn}¤
  °Lock l, State s, Lock r) {°
  °if (l == Free && r == l && s == Try) {°¤\label{Line:Overview:PcmcSoln:BaseCase}¤
  °  return true;°
  °} return synth(l, s, r); }°
void main(struct View v) {
  if (*) { ¤\label{Line:Overview:PcmcSoln:If}¤
    ^State otr = *;^¤\label{Line:Overview:PcmcSoln:ClosureSetupStart}¤
    ^assume(RInv(v.left, v.st, v.right));^
    ^assume(RInv(v.right, otr, v.left));^ ¤\label{Line:Overview:PcmcSoln:ClosureSetupEnd}¤
    v = tr(v); ¤\label{Line:Overview:PcmcSoln:ClosureTr}¤
    ~assert(RInv(v.left, v.st, v.right));~ ¤\label{Line:Overview:PcmcSoln:ClosureCheckStart}¤
    ~assert(RInv(v.right, otr, v.left));~ ¤\label{Line:Overview:PcmcSoln:ClosureCheckEnd}¤
  } else {
    ^assume(RInv(v.left, v.st, v.right));^ ¤\label{Line:Overview:PcmcSoln:AdequateSetup}¤
    bool held = v.left == Left && ¤\label{Line:Overview:PcmcSoln:AdequateCheckStart}¤
                v.right == Right;
    assert(v.st != Critical || held); }} ¤\label{Line:Overview:PcmcSoln:AdequateCheckEnd}¤
\end{lstlisting}
    \caption{The \ipsmp problem (uses \code{tr}).}
    \label{Fig:Overview:PcmcSoln}
  \end{subfigure}
   \caption{A process for a parameterized ring, and an \ipsmp problem that verifies the process. The process is correct relative to the compositional invariant $( ( \textsf{v.lhs} \ne \textsf{Left} ) \lor ( \textsf{v.rhs} \ne \textsf{Right} ) ) \implies ( \textsf{v.st} \ne \textsf{Critical} )$, and the \ipsmp problem synthesizes the compositional invariant. Note that \code{Lock} and \code{State} are defined in \cref{Fig:Overview:PcmcProb} using \code{typedef}, and that \code{otr} is a process under interference.}
\end{figure}

As a third example, we illustrate a reduction from parameterized verification to \ipsmp.
This example considers two or more processes running in a ring network of arbitrary size.
A ring network organizes processes into a single cycle, such that each process has a left and right neighbour~\cite{ChandyMisra1989}.
In this ring, adjacent processes share a lock on a common resource.
Processes are either \emph{trying} to acquire a lock, or have acquired all locks and are in a \emph{critical} section.
Initially, all processes are trying and all locks are free.
Each processes runs the program in \cref{Fig:Overview:PcmcProb}.
The state of each process is given by \code{View} on \auxlineref{Line:Overview:PcmcProb:View}, and the transition relation of each process is given by \code{tr}\footnote{For simplicity, \code{tr} is not deadlock-free as processes retain their locks until reaching their critical sections. However, the critical section can be reached any number of times without encountering a deadlock.} on \auxlineref{Line:Overview:PcmcProb:Tr}.
Since each process runs the same program with the same configuration of locks, the ring network is said to be \emph{symmetric}.

The goal of this example is to prove that if a process is in its critical section, then the process holds both adjacent locks.
Following \cite{NamjoshiTrefler2016}, this property is proven by computing an adequate compositional invariant for each process.
An adequate compositional invariant is true of the initial state of each process, closed under the transition relation, closed under interference, and entails the property of interest.
Remarkably, in a symmetric ring network, a compositional invariant can be computed by analyzing a ring with exactly two processes.

Using this approach, the program in \cref{Fig:Overview:PcmcSoln} is obtained.
This program uses a non-deterministic \code{if} statement to split the analysis into two cases (\auxlineref{Line:Overview:PcmcSoln:If}).
The first case instantiates a two-process network using the compositional invariant (\linerange{Line:Overview:PcmcSoln:ClosureSetupStart}{Line:Overview:PcmcSoln:ClosureSetupEnd}).
Due to network symmetry, the left lock of the first process is the right lock of the second process, and vice versa.
A single process in this network executes a transition (\auxlineref{Line:Overview:PcmcSoln:ClosureTr}), and then the closure of the compositional invariant is asserted for both processes (\linerange{Line:Overview:PcmcSoln:ClosureCheckStart}{Line:Overview:PcmcSoln:ClosureCheckEnd}).
The assertions on \linerange{Line:Overview:PcmcSoln:ClosureCheckStart}{Line:Overview:PcmcSoln:ClosureCheckEnd} ensure both closure under the transition relation and closure under interference, since only a single process transitioned.
The second case instantiates a single process using the compositional invariant (\auxlineref{Line:Overview:PcmcSoln:AdequateSetup}), and then asserts the property of interest (\linerange{Line:Overview:PcmcSoln:AdequateCheckStart}{Line:Overview:PcmcSoln:AdequateCheckEnd}).
Together, these cases define a compositional invariant.
This gives a program with unknowns, as required by the verification methodology.

Next, the shape of the compositional invariant is considered.
In this example, there is no motivation to split the invariant into cases.
Furthermore, it would not make sense to partition the arguments of the invariant, since the state of a process is dependent on the combined state of its adjacent locks.
Therefore, $\code{RInv}( l, s, r )$ is assumed to be the shape of the invariant.
In the \ipsmp encoding, \code{RInv} corresponds to the partial predicate on \auxlineref{Line:Overview:PcmcSoln:Pfn}.
The template on \auxlineref{Line:Overview:PcmcSoln:BaseCase} ensures that the compositional invariant is true of the initial state of each process.
As an alternative to a template, one can instead assert \code{RInv(Free,Try,Free)} before \auxlineref{Line:Overview:PcmcSoln:If}.
One solution to this problem is to fill the hole in \code{RInv} with the expression $\code{(s == Try) || ((l == Left) \&\& (r == Right))}$.
Consequently, $( ( s \ne \code{Try} ) \implies ( l = \code{Left} \land r = \code{Right} ) )$.

\subsection{Discussion}
\label{Sect:Overview:Discussion}

\begin{figure}[t]
  \begin{subfigure}{0.47\textwidth}
\begin{lstlisting}[style=ipsmp]
°bool PRED_TEMPLATE Inv(int x, int y) {°
  °if (x + y == 5) { return true; }°
  °else { return synth(x, y); } }°
void main(void) { ... }
\end{lstlisting}
    \caption{Predicate template encoding.}
    \label{Fig:Overview:Bound}
  \end{subfigure}
  \hfill
  \begin{subfigure}{0.47\textwidth}
\begin{lstlisting}[style=ipsmp]
°bool PRED_TEMPLATE Inv(int x, int y) {°
  °return synth(x, y); }°
void main(void) {
  int x = *; int y = *;
  assume(x + y == 5);
  ~assert(Inv(x, y));~
  ... }
\end{lstlisting}
    \caption{Assertion-based encoding.}
    \label{Fig:Overview:Assert}
  \end{subfigure}
  \caption{The initial condition $( x + y = 5 )$ encoded using a predicate (see \cref{Fig:Overview:Bound}), and its equivalent encoding using an assertion (see \cref{Fig:Overview:Assert}).}
\end{figure}

In Sec.~\ref{Sect:Overview:Array}, all explicit constraints were easily replaced by implicit constraints.
However, explicit constraints can yield more succinct encodings.
For example, consider the initial condition $( x + y = 5 )$.
In \cref{Fig:Overview:Bound}, the condition is given as an explicit predicate template, and in \cref{Fig:Overview:Assert}, it is desugared as an assertion.
To desugar the constraint, additional variables and assumptions are required.

In the examples presented so far, each \ipsmp problem places both assumptions and assertions on each partial predicate.
All interesting \ipsmp problems follow this pattern.
However, \ipsmp is well defined even if a partial predicate has only assumptions placed on it, only assertions placed on it, or neither.
In these cases, the \ipsmp problem is trivial or reduces to a simpler problem.

If partial predicates only appear in assumptions, then the synthesized solution is never strengthened.
In other words, the solution may be arbitrarily weak.
This is an instance of specification synthesis.
Usually, in specification synthesis, non-functional requirements are placed on each specification to ensure that a solution is ``interesting'' (e.g.,~\cite{DasLahiri2015,AlbarghouthiDillig2016,PrabhuFedyukovich2021}).
Without these requirements, uninteresting solutions, such as \code{false}, are permitted.
Since \ipsmp only places functional requirements on its solutions, this case is trivial and returning \code{false} from each predicate is always a solution (given a correct program).

If partial predicates only appear in assertions, then the synthesized solution is only ever strengthened.
A solution in this case is an expression that subsumes all assertions placed on the predicate.
However, an expression that evaluates to \code{true} subsumes all possible assertions.
Therefore, this case is also trivial and returning \code{true} from each predicate is always a solution (given a correct program).

If partial predicates never appear in the program, then the synthesizer can select an arbitrary implementation for each predicate.
However, if the synthesizer returns a solution, then the program must be correct relative to the solution.
Therefore, if the program does violate an assertion, then the synthesizer must return a witness to unrealizability instead.
Consequently, the output of the synthesizer indicates if the program is correct, and is equivalent to verification.

%% file: background.tex
\section{Background}
\label{Sect:Background}

This section recalls results from logic-based program verification.
\cref{Sect:Background:FOL} reviews the key definitions of First Order Logic (FOL) and the Constrained Horn Clause (CHC) fragment of FOL. 
\cref{Sect:Background:Lang} introduces a programming language used throughout this paper.
\cref{Sect:Background:WLP} recalls the connection between CHCs and program semantics through the weakest liberal precondition transformer.

\subsection{First Order Logic and Constrained Horn Clauses}
\label{Sect:Background:FOL}

\newcommand{\SigmaBool}{\Sigma_{\mathrm{Bool}}}

\newcommand{\cV}{\mathcal{V}}
\newcommand{\Terms}{\mathrm{Term}}
\newcommand{\QFFml}{\mathrm{QFFml}}
A \emph{first order signature} $\Sigma$ defines a set of predicates, a set of relations, and their respective arities.
Given a set of variables $\cV$, a \emph{term} is either a variable from $\cV$ or an application of a relation in $\Sigma$ to one or more terms.
An \emph{atom} is an application of a predicate in $\Sigma$ to one or more terms.
A \emph{formula} joins atoms using standard logical connectives, existential quantification, and universal quantification.
A formula is \emph{quantifier-free} if it contains neither existential nor universal quantification.
A formula is a \emph{sentence} if all variable instances are quantified.
Given a FO-formula $\varphi$, the formula $\varphi[ x / y ]$ is defined by substituting $y$ for all free instances of $x$ in $\varphi$.
We write $\Terms( \Sigma, \cV )$ and $\QFFml( \Sigma, \cV )$ for the sets of terms and quantifier-free formulas generated by $\Sigma$ and $\cV$.

\newcommand{\cF}{\mathcal{F}}
\newcommand{\cT}{\mathcal{T}}
A FO-\emph{theory} $\cT$ is a deductively closed set of sentences over a signature $\Sigma$.
A \emph{$\cT$-model} for a formula $\varphi$ is an interpretation of each predicate, relation, and free variable in $\cT \cup \{ \varphi \}$ such that every formula in $\cT \cup \{ \varphi \}$ is true.
If a $\cT$-model exists for $\varphi$, then $\varphi$ is \emph{satisfiable}, otherwise, $\varphi$ is \emph{unsatisfiable}.
In the case that all valid interpretations of $\cT$ are $\cT$-models for $\varphi$, then $\varphi$ is $\cT$-valid and we write $\models_{\cT} \varphi$.
Furthermore, if each interpretation of a $\cT$-model $M$ can be expressed in some logical fragment $\cF$, then $M$ provides an \emph{$\cF$-solution} to $\varphi$ .

\newcommand{\keep}{\textsf{keep}}
Constrained Horn Clauses (CHCs) are a fragment of FOL determined by a FO-signature $\Sigma$ and an set of predicates $P$.
A CHC is a sentence of the form
$\forall \,V \cdot \varphi \land p_1( \vec{x}_1 ) \land \cdots \land p_k( \vec{x}_k ) \implies h( \vec{y} )$, 
where $\varphi \in \QFFml( \Sigma, \cV )$ and $\{ p_1, \ldots, p_k, h \} \subseteq P$.
For program semantics, it is useful to use $v'$ to denote the value of a variable $v$ after a program transition and $\keep( W ) := \bigwedge_{w \in W} w = w'$ to denote that variables $W \subseteq \cV$ are unchanged during a transition.
Given a set of variables $V = \{ v_1, \ldots, v_n \} \subseteq \cV$, the set of variables $\{ v_1', \ldots, v_n' \}$ is denoted $V'$.
Likewise, given a formula $\varphi$ over the variables in $V$, the formula $\varphi[v_1/v_1']\cdots[v_n/v_n']$ over $V'$ is denoted $\varphi'$.

\subsection{Procedural Programming Language}
\label{Sect:Background:Lang}

\newcommand{\Programs}{\textsf{Progs}}

\newcommand{\cP}{\mathcal{P}}
\newcommand{\Procs}{\textsf{Procs}}
Throughout this paper, we consider a simple procedural programming language, whose syntax is standard and can be found in~\appendixcite{Appendix:Language}.
We assume that all expressions are factored out by a FO-signature $\Sigma$, with variables from a set $\cV$.
That is, each expression is of the form $\QFFml( \Sigma, \cV )$.
The set of all programs in the language is denoted $\Programs( \Sigma, \cV )$.
For simplicity, types are omitted.
In this language, a program has one or more procedures, with execution starting from \code{main}.
Each procedure is written in an imperative language, including loops and procedure calls.
The language is extended with a non-deterministic assignment (i.e.,~\code{*}), and verification statements \code{assume} and \code{assert}.
The expressions in assume and assert can be either from $\QFFml( \Sigma, \cV )$ or a call to a pure Boolean procedure, called a \emph{predicate}.
Predicates may only be called within assume or assert statements.
Given a program $\cP \in \Programs( \Sigma, \cV )$, $\Procs( \cP )$ denotes the procedures in $\cP$.
A special case is when all variables are Boolean.

\newcommand{\Locs}{\mathit{Locs}}
\newcommand{\GV}{\mathit{GV}}
\newcommand{\LV}{\mathit{LV}}
\newcommand{\Bops}{\mathit{NE}}
\newcommand{\Bcalls}{\mathit{CE}}
\newcommand{\Bassume}{\mathit{FE}}
\newcommand{\Bassert}{\mathit{AE}}
\newcommand{\Bsum}{\mathit{PE}}
\newcommand{\lcall}{l_{\mathrm{call}}}
\newcommand{\lin}{l_{\mathrm{in}}}
\newcommand{\lout}{l_{\mathrm{out}}}
\newcommand{\lret}{l_{\mathrm{ret}}}
\begin{definition}
  \label{Def:BoolProg}
  Let $\SigmaBool$ denote a Boolean signature.
  A \emph{Boolean program} is a tuple $( \Locs, \GV, \LV, E )$ with $E = \left( \Bops, \Bcalls, \Bassume, \Bassert, \Bsum \right)$ and $V = \GV \cup \LV$ such that:
  \begin{enumerate}
  \item $\Locs$ is a finite set of control-flow locations with entry-point $\code{main} \in \Locs$;
  \item $\GV$ and $\LV$ are disjoint sets of local and global variables (respectively);
  \item $\Bops \subseteq \Locs \times \QFFml \left( \SigmaBool, V \cup V' \right) \times \Locs$ is a set of \emph{normal edges}, $\Bcalls \subseteq \Locs \times \Locs \times \Locs$ is a set of \emph{call edges}, $\Bassume \subseteq \Locs \times \Locs \times \Locs$ is a set of \emph{(partial predicate) call-under-assume edges}, $\Bassert \subseteq \Locs \times \Locs \times \Locs$ is a set of \emph{(partial predicate) call-under-assert edges}, and $\Bsum \subseteq \Locs \times \Locs$ is a set of \emph{procedure summary edges};
  \item If $( l_1, R, l_2 ) \in \Bops$, then $l_2$ is reachable from $l_1$ by updating the variables according to $R$ and if $( \lcall, \lin, \lret ) \in \Bcalls$, then $\lret$ is reachable from $\lcall$ by executing the procedure with entry location $\lin$;
  \item If $( \lcall, \lin, \lret ) \in \Bassume$, then $\lret$ is reachable from $\lin$ by assuming the partial predicate with entry location $\lcall$ and if $( \lcall, \lin, \lret ) \in \Bassert$, then $\lret$ is reachable from $\lin$ by asserting the partial predicate with entry location $\lcall$;
  \item If $( \lin, \lout ) \in \Bsum$, then the procedure with entry location $\lin$ has exit location $\lout$.
  \end{enumerate}
\end{definition}
      
\newcommand{\lbot}{l_{\bot}}
A Boolean program consists of control-flow locations and edges between locations.
Each procedure has a single entry location, $\lin$, and a single exit location, $\lout$, where $( \lin, \lout ) \in \Bsum$.
The program enters at $\code{main} \in \Locs$, and a special location $\lbot \in \Locs$ indicates failure.
\emph{Normal edges} connect control-flow locations within a procedure and represent non-procedural statements.
For example, the statement \code{assert(e)} (where $e$ is an expression) corresponds to two normal edges, $\left( l_1, e \land \keep( \GV \cup \LV ), l_2 \right)$ and $\left( l_1, \neg e \land \keep( \GV \cup \LV ), \lbot \right)$.
\emph{Call edges} (optionally under assume or assert) connect locations in a caller's procedure and a callee's procedure by giving the call and return locations of the caller ($\lcall$ and $\lret$, respectively), and the entry location for the callee ($\lin$).
For simplicity, all procedures have the same local variables, and arguments are passed by global variables.
The location $\lbot$ is assumed to have no outgoing edges.

The state of a Boolean program is a tuple $( l, s )$, where $l$ is a location and $s$ is an assignment to each Boolean variable.
Initially, $l = \code{main}$ and $s$ is an arbitrary assignment.
For each normal edge $( l_1, R, l_2 )$, a transition exists from $( l_1, s_1 )$ to $( l_2, s_ 2 )$, if $s_1 \land R \land s_2'$ is valid.
All call edges have the expected semantics.

\subsection{Logical Program Verification}
\label{Sect:Background:WLP}

\newcommand{\Funcs}{\textsf{Funcs}}
\newcommand{\wlp}{\mathit{wlp}}
\newcommand{\ToCHC}{\mathsf{ToCHC}}
\newcommand{\linv}{\mathit{loop}}
\newcommand{\wlpassert}{\textsf{assert}}
\newcommand{\wlpassume}{\textsf{assume}}
\newcommand{\wlpskip}{\textsf{skip}}
\newcommand{\wlpif}{\textsf{if}}
\newcommand{\wlpelse}{\textsf{else}}
\newcommand{\wlpwhile}{\textsf{while}}
\newcommand{\wlpret}{\textsf{return}}

The Weakest Liberal Precondition (WLP) transformer gives logical semantics to imperative programs~\cite{Dijkstra1975}.
We write $\wlp( S, Q )$ for the WLP of a statement $S$ with respect to a post-condition$Q$.
The WLP transformer for $\Programs( \Sigma, \cV )$ is standard and can be found in~\appendixcite{Appendix:Language}.
Note that in this transformation, the $\linv_{ln}$ predicate is an invariant for a loop at line $\textit{ln}$.

\begin{figure}[t]
  \input{diagrams/tochc}
  \vspace{-0.25in}
  \caption{The partial correctness conditions for a program $\cP \in \Programs( \Sigma, \cV )$. This follows the presentation of~\cite{BjornerGurfinkel2015}.}
  \label{Fig:Background:ToCHC}
\end{figure}

The $\wlp( - )$ transformer can be used to verify partial correctness for procedural programs.
This is achieved through the $\ToCHC( - )$ transformer in \cref{Fig:Background:ToCHC}.
The $\wlp( \cP( \code{main} ), \top )$ term asserts that \code{main} satisfies all assertions.
For each procedure $f \in \Procs( \cP )$, the term $\ToCHC( f )$ asserts that $f$ is correct for all inputs passed to $f$ in every execution.
Note that in \cref{Fig:Background:ToCHC}, $f_{pre}$ collects inputs to $f$, and $f_{sum}$ relates the inputs of $f$ to the outputs of $f$.
In the case that $f$ is a predicate, $f_{pre}$ and $f_{sum}$ are omitted, since $f$ is side-effect free.
Together, $\ToCHC( \cP )$ asserts that the program $\cP$ is correct for any execution starting from \code{main}.
If $\ToCHC( \cP )$ is satisfiable, then there exist loop invariants for $\cP$ such that $\cP$ satisfies all assertions~\cite{BjornerGurfinkel2015}.
Therefore, $\ToCHC( \cP )$ can be used to verify $\cP$.
Furthermore, $\ToCHC( \cP )$ is in the CHC fragment~\cite{BjornerGurfinkel2015}.

\newcommand{\bddexists}{\textsf{elim}}
Efficient procedures exist to prove that Boolean programs are correct.
For example, \emph{program summarization} simultaneously computes a summary $\theta$ from control-flow locations to input-to-reachable-state relations, and a summary $\sigma$ from procedures to input-output relations.
For a location $l \in \Locs$, if $\theta( l ) = \bot$, then $l$ is unreachable.
Therefore, a Boolean program $\cP$ is correct if and only if $\theta( \lbot ) = \bot$ in the least summary of $\cP$~\cite{AlurBouajjani2018}.
Program summarization is defined in \cref{Def:Summary}\footnote{To align with \cref{Alg:BebopReach}, \cref{Def:Summary} is non-standard but equivalent to \cite{AlurBouajjani2018}.}.
The algorithm to compute $\theta$ is presented in full, for reuse in \cref{Sect:Decidability:Bool}.
For presentation, $\bddexists( \varphi, W )$ denotes the existential elimination of $W$ in $\varphi$.

\begin{definition}[\cite{AlurBouajjani2018}]
  \label{Def:Summary}
  A \emph{Boolean program summary} for $\left( \Locs, \GV, \LV, E \right)$, where $E = \left( \Bops, \Bcalls, \varnothing, \varnothing, \Bsum \right)$ is a tuple $( \theta, \sigma )$ such that $Q = \QFFml( \SigmaBool, V \cup V' )$, $V = \GV \cup \LV$ and the following hold:
  \begin{enumerate}
  \item $\sigma: \Locs \rightarrow Q$ and $\theta: \Locs \rightarrow Q$;
  \item $\sigma \left( \code{main} \right) = \top$;
  \item $\forall \left( l_1, R, l_2 \right) \in \Bops \cdot \theta \left( l_1 \right) \land R' \implies \theta \left( l_2 \right)[ V' / V'' ]$;
  \item $\forall \left( \lcall, \lin, \lret \right) \in \Bcalls \cdot \theta \left( \lcall \right) \land \sigma' \left( \lin \right) \land \keep \left( \LV' \right) \implies \theta \left( \lret \right)[ V' / V'' ]$;
  \item $\forall \left( \lcall, \lin, \lret \right) \in \Bcalls \cdot \bddexists( \theta \left( \lcall \right), V \cup \LV') \implies \theta \left( \lin \right)$;
  \item $\forall \left( \lin, \lout \right) \in \Bsum \cdot \theta \left( \lout \right) \implies \sigma \left( \lin \right)$.
  \end{enumerate}
\end{definition}

\code{ComputeBoolReach} in \cref{Alg:BebopReach} is the standard algorithm to compute a least program summary.
The algorithm works by iteratively applying the rules of \cref{Def:Summary} until a fixed point is reached (we write $R^{*} := R[V'/V''][V/V']$).
Termination is ensured by the finite-state of Boolean programs and the monotonicity of each rule.
We extend on the algorithm \code{ComputeBoolReach} in \cref{Sect:Decidability:Bool}.

\newcommand{\lwork}{l_{\mathrm{wk}}}
\newcommand{\vwork}{s_{\mathrm{wk}}}
\newcommand{\vsum}{s_{\mathrm{sum}}}
\newcommand{\vret}{s_{\mathrm{ret}}}
\newcommand{\vin}{s_{\mathrm{in}}}
\newcommand{\vdiff}{s_{\mathrm{diff}}}
\newcommand{\worklist}{\mathrm{W}}
\newcommand{\var}{\mathbf{var}}
\begin{algorithm}[t]
  \caption{Computes the least Boolean program summary $( \theta, \sigma )$~\cite{BallRajamani2001}. Follows the presentation of~\cite{ChakiGurfinkel2018}.}\label{Alg:BebopReach}
  \footnotesize
  \setlength{\columnsep}{0.65cm}
  \begin{multicols}{2}
  \input{pseudocode/bebop}
  \end{multicols}
  \normalsize
  \vspace{0.5em}
\end{algorithm}

%% file: diagrams/tochc.tex
\renewcommand{\gets}{\mathbin{:=}}
\begin{gather*}
  \begin{aligned}
  \ToCHC( \cP ) &\gets \wlp( \cP( \texttt{Main} ), \top ) \land \left( \bigwedge_{f \in \Procs( \cP )} \ToCHC( f ) \right)\\
  \ToCHC \left( f( \vec{x} )\, \{\; S;\; \wlpret\; \vec{e};\; \} \right) &\gets
    \forall \overline{\vec{x}} \cdot \left(
      \vec{x} = \overline{\vec{x}} \land f_{pre}( \vec{x} )
      \implies \wlp\left( S, f_{sum} \left( \overline{\vec{x}}, \vec{e} \right) \right)
    \right) \\
  \ToCHC \left( p( \vec{x} )\, \{\; \wlpret\; e;\; \} \right) &\gets
    \forall \vec{x} \cdot p( \vec{x} ) \iff e \\
  \end{aligned}
\end{gather*}

%% file: pseudocode/bebop.tex
\SetKwFunction{FDoIntraprocedural}{DoIntraproc}
\SetKwFunction{FDoProcedures}{DoProcs}
\SetKwFunction{FDoProcedureSummary}{DoProcSum}
\SetKwFunction{FUpdateReach}{UpdateReach}
\SetKwFunction{FInitBoolReach}{InitBoolReach}
\SetKwFunction{FComputeBoolReach}{ComputeBoolReach}
\SetKwProg{Fn}{Func}{:}{}
$\var\;( \theta, \sigma )$
\tcp*{A program summary}
$\var\; \worklist$
\tcp*{A map from $\Locs$ to a queued state}

\Fn{\FUpdateReach{$l$, $v$}}{
  $\vdiff \gets v \land \neg\, \theta( l )$\;
  \uIf{$\vdiff \ne \bot$}{
    $\theta( l ) \gets \theta( l ) \lor \vdiff$\;
    $\worklist( l ) \gets \worklist( l ) \lor \vdiff$
  }
}

\Fn{\FDoIntraprocedural{$V$, $\Bops$, $\lwork$, $\vwork$}}{
  \For{$( \lwork, R, l_2 ) \in \Bops$}{
    $s_2 \gets \bddexists( \vwork \land R^*, V' )$\;
    $s_2 \gets s_2[ V'' / V' ]$\;
    \FUpdateReach{$l_2$, $s_2$}\;
  }
}

\Fn{\FDoProcedureSummary{$V$, $\LV$, $\Bsum$, $\Bcalls$, $\lwork$, $\vwork$}}{
  \For{$( \lin, \lwork ) \in \Bsum$}{
    $\vsum \gets \bddexists( \vwork, \LV \cup \LV' ) \land \neg \sigma( \lin )$\;
    \lIf{$\vsum = \bot$}{$\textsf{continue}$}
    $\sigma( \lin ) \gets \sigma( \lin ) \lor \vsum$\;
    \For{$(\lcall, \lin, \lret ) \in \Bcalls$}{
      $X \gets \vsum^* \land \keep( \LV' )$ \;
      $s \gets \bddexists( \theta( \lcall ) \land X, V' )[ V'' / V' ]$\;
      \FUpdateReach{$\lret$, $s$}\;
    }
  }
}

\Fn{\FDoProcedures{$V$, $\LV$, $\Bsum$, $\Bcalls$, $\lwork$, $\vwork$}}{
  \For{$( \lwork, \lin, \lret ) \in \Bcalls$}{
    $\vin \gets \bddexists( \vwork, V \cup \LV' )[ V' / V ]$\;
    \FUpdateReach{$\lin$, $\vin$}\;
    $X \gets \sigma( \lin )^* \land \keep( \LV' )$\;
    $s \gets \bddexists( \vwork \land X, V' )[ V'' / V' ]$\;
    \FUpdateReach{$\lret$, $s$}\;
  }
}

\Fn{\FInitBoolReach{$\Locs$, $\Bsum$}}{
  \lFor{$l \in \Locs$}{
    $\theta( l ) \gets \bot$
  }
  \lFor{$( \lin, \lout ) \in \Bsum$}{
    $\sigma( \lin ) \gets \bot$
  }
  $\theta( \texttt{main} ) \gets \top$;
  $\worklist( \texttt{main} ) \gets \top$\;
}

\Fn{\FComputeBoolReach{$\cP$}}{
  $( \Locs, \GV, \LV, ( N, C, \varnothing, \varnothing, P ) ) \gets \cP$\;
  $V \gets \GV \cup \LV$\;
  \FInitBoolReach{$\Locs$, $P$}\;
  \While{$\exists \;\lwork \in \Locs \cdot \worklist( \lwork ) \ne \bot$}{
    $\vwork \gets \worklist( \lwork )$ ;
    $\worklist( \lwork ) \gets \bot$\;
    \FDoIntraprocedural{$V$, $N$, $\lwork$, $\vwork$}\;
    \FDoProcedures{$V$, $\LV$, $P$, $C$, $\lwork$, $\vwork$}\;
    \FDoProcedureSummary{$V$, $\LV$, $P$, $C$, $\lwork$, $\vwork$}\;
  }
}

%% file: synthesis.tex
\section{IPS-MP: Problem Definition}
\label{Sect:Problem}

\newcommand{\Partial}{\textsf{\textsf{Partial}}}
This section defines partial predicates and the \ipsmp problem.
A \emph{partial predicate} is a pure Boolean function without an implementation.
A program $\cP$ is \emph{open} if it contains a partial predicate $p$.
An implementation for $p$ is a Boolean expression $e$ over the arguments of $p$.
The program obtained by implementing $p$ as $\code{return}\; e$ is denoted $\cP[ p \leftarrow e ]$.
The set of all partial predicates in $\cP$ is written $\Partial( \cP ) = \{ p_1, \ldots, p_k \}$.
Given a function $\Pi$ from $\Partial( \cP )$ to pure Boolean expressions, we write $\cP[ \Pi ]$ to denote \mbox{$\cP[ p_1 \leftarrow \Pi( p_1 ) ] \cdots [ p_k \leftarrow \Pi( p_k ) ]$}.
The \ipsmp problem is to find a $\Pi$ such that $\cP[\Pi]$ is correct.

\begin{figure}[t]
  \begin{subfigure}{0.48\textwidth}
\begin{lstlisting}[style=ipsmp]
bool Post(int x, int y) {
  return x==y; }
void main(int y) {
  assume(y>0); int x=0;
  for (int i=0; i<y; ++i) { x+=1; }
  assert(Post(x, y));
  x=*; y=*; assume(Post(x, y)); }
\end{lstlisting}
    \caption{The program $\cP[\Pi_{\mathit{weak}}]$.}
    \label{Fig:Synth:Weakest}
  \end{subfigure}
  \hfill
  \begin{subfigure}{0.48\textwidth}
\begin{lstlisting}[style=ipsmp]
bool Post(int x, int y) {
  return (y>0) && (x==y); }
void main(int y) {
  assume(y > 0); int x=0;
  for (int i=0; i<y; ++i) { x+=1; }
  assert(Post(x, y));
  x=*; y=*; assume(Post(x, y)); }
\end{lstlisting}
    \caption{The program $\cP[\Pi_{\mathit{strong}}]$.}
    \label{Fig:Synth:Strongest}
  \end{subfigure}
  \caption{Implementations of the simple program in \cref{Fig:Overview:Sample}.}
\end{figure}

\begin{example}
  Recall program $\cP$ from \cref{Fig:Overview:Sample}.
  Since \code{Post} is unimplemented in $\cP$, then $\cP$ is an open program.
  Formally, $\Partial( \cP ) = \{ \code{Post} \}$.
  In \cref{Sect:Overview}, two implementations were proposed for \code{Post}, namely $( x = y )$ and $( y > 0 \land x = y )$.
  These implementations are represented by the mappings $\Pi_{\mathit{weak}}$ and $\Pi_{\mathit{strong}}$ from $\Partial( \cP )$ to pure Boolean expressions such that $\Pi_{\mathit{weak}}: \code{Post} \mapsto ( x = y )$ and $\Pi_{\mathit{strong}}: \code{Post} \mapsto ( y > 0 \land x = y )$.
  The closed programs $\cP[\Pi_{\mathit{weak}}]$ and $\cP[\Pi_{\mathit{strong}}]$ are illustrated in \cref{Fig:Synth:Weakest} and \cref{Fig:Synth:Strongest}, respectively.
  \qed
\end{example}

\begin{definition}
  \label{Def:IPSMP}
  An \emph{\ipsmpfull~(\ipsmp)} problem is a tuple $( \cP, \cT, \Pi_0 )$ such that $\cP \in \Programs( \Sigma, \cV )$ with first-order signature $\Sigma$ and variable set $\cV$, $\cT$ is a first-order theory, and $\Pi_0: \Partial( \cP ) \rightarrow \QFFml( \Sigma, \cV )$ are predicate templates.
  A \emph{solution} to $( \cP, \cT, \Pi_0 )$ is a function $\Pi: \Partial( \cP ) \rightarrow \QFFml( \Sigma, \cV )$ such that $\cP[ \Pi ]$ is correct relative to $\cT$ and $\forall p \in \Partial( \cP ) \,\cdot \models_{\cT} \Pi_0( p ) \implies \Pi( p )$.
\end{definition}

Assume that $( \cP, \cT, \Pi_0 )$ is an \ipsmp problem with a solution $\Pi$.
With respect to the \ipsmp overview in \cref{Fig:Overview:Pipeline}, $\cP$ is a \emph{program with specifications}, $\Pi_0$ is a collection of \emph{predicate templates}, and $\Pi$ is an \emph{implementation of partial predicates}.
The \emph{witness to unrealizability} is discussed in \cref{Sect:Decidability}.
As an example of \cref{Def:IPSMP}, \cref{Fig:Overview:ArraySoln} is restated as a formal \ipsmp problem.

\begin{example}
  This example restates \cref{Fig:Overview:ArraySoln} as an \ipsmp problem $( \cP, \Pi_0, \cT )$.
  The program $\cP$ is given by \linerange{Line:Overview:ArraySoln:MainIn}{Line:Overview:ArraySoln:MainOut} of \cref{Fig:Overview:ArraySoln}.
  Then $\Partial( \cP ) = \{ \code{Inv3}, \code{Inv4} \}$, since \code{Inv3} and \code{Inv4} are called on \linerange{Line:Overview:ArraySoln:Inv3Called}{Line:Overview:ArraySoln:Inv4Called}, but lack full implementations.
  From \linerange{Line:Overview:ArraySoln:Pfn1}{Line:Overview:ArraySoln:Pfn2End}, $\Pi_0( \code{Inv3} ) = \Pi_0( \code{Inv4} ) = ( m = 0 \land v = 0 )$.
  Now, recall from \cref{Sect:Overview:Array} that all variables in \cref{Fig:Overview:ArraySoln} are arithmetic integers.
  Therefore, $\cT$ is the theory of integer linear arithmetic.
  A solution to $( \cP, \Pi_0, \cT )$ is $\Pi$ such that $\Pi( \code{Inv3} ) = ( v = 0 )$ and $\Pi( \code{Inv4} ) = ( 0 \le v \land v \le m )$.
  \qed
\end{example}

%% file: decidability.tex
\section{Decidability of IPS-MP}
\label{Sect:Decidability}

This section considers the decidability of \ipsmp.
\cref{Sect:Decidability:Bool} shows that \ipsmp is efficiently decidable in the Boolean case.
\cref{Sect:Decidability:General} shows that \ipsmp is undecidable in general, but admits sound proof-rules for realizability and unrealizabiliy.

\subsection{The Case of Boolean Programs}
\label{Sect:Decidability:Bool}

This section shows that for Boolean programs, \ipsmp is decidable with the same time complexity as problem verification (i.e.,~polynomial in the number of program states).
In contrast, general synthesis is known to have exponential time complexity in the Boolean case~\cite{Vardi2008}.
Therefore, \ipsmp modulo Boolean programs does in fact offer the benefits of general synthesis without the associated costs.
To prove this result, we first extend Boolean program summaries (\cref{Def:Summary}) to programs with partial predicates.
These new summaries are then used to extract solutions to \ipsmp (or witnesses to unrealizability).
\code{Analyze} of \cref{Alg:BoolSynth} extends on \cref{Alg:BebopReach} to compute these new summaries.
The total correctness and time complexity of \code{Analyze} are proven in \cref{Cor:BoolCorrectness} and \cref{Thm:BoolComplexity}, respectively.

To simplify our presentation, we assume that all predicates are partial.
In a Boolean program, each partial predicate has an entry location, but no edges nor exit location.
This means that a standard summary can be obtained for a Boolean program with partial predicates by discarding all calls to partial predicates.
Such a summary characterizes reachability, under the assumption that partial predicates are never called.
From this summary, the arguments passed to each partial predicate under assert can be collected.
For the program summary to be correct, the partial predicates must return true on these asserted arguments.
If the partial predicate returns true on these asserted arguments, then for any call under assume using the same arguments, the program execution must continue to the next state.
This procedure can then be repeated until a fixed point is obtained.
This new \emph{partial program summary} is defined formally in \cref{Def:PartialSummary}.
\begin{definition}
  \label{Def:PartialSummary}
  Let $\cP =  ( \Locs, \GV, \LV, ( \Bops, \Bcalls, \Bassume, \Bassert, \Bsum ) )$ be a Boolean program.
  A \emph{partial program summary} for $\cP$ is a tuple $( \theta, \sigma, \Pi )$ such that:
  \begin{enumerate}
  \item $\Pi: \Partial( \cP ) \rightarrow \QFFml( \SigmaBool, \GV )$;
  \item $( \theta, \sigma )$ is a program summary for $\left( \Locs, \GV, \LV, \left( \Bops, \Bcalls, \varnothing, \varnothing, \Bsum \right) \right)$;
  \item $\forall ( \lcall, \lin, \lret ) \in \Bassert \cdot \theta( \lcall ) \implies \Pi'( \lin )$;
  \item $\forall ( \lcall, \lin, \lret ) \in \Bassert \cdot \theta( \lcall ) \implies \theta( \lret )$;
  \item $\forall ( \lcall, \lin, \lret ) \in \Bassume \cdot \theta( \lcall ) \land \Pi'( \lin ) \implies \theta( \lret )$.
  \end{enumerate}
\end{definition}

The rules of \cref{Def:PartialSummary} follow directly from the preceding discussion.
Rule~2 ensures that $( \theta, \sigma )$ is a program summary for $\cP$ after discarding all calls to partial predicates.
Rules3~and~4 collect the arguments passed to partial predicates under assert.
Rule~5 advances the program state from calls to partial predicates under assume, according to the collected arguments.
These steps are made operational by \code{Analyze} of \cref{Alg:BoolSynth}.
Note that \code{Analyze} does not call \code{ComputeBoolReach} directly, and instead applies all rules in a single loop.

\begin{algorithm}[t]
  \caption{An extension of \cref{Alg:BebopReach} to solve \ipsmp for Boolean programs.}\label{Alg:BoolSynth}
  \footnotesize
  \setlength{\columnsep}{0.65cm}
  \begin{multicols}{2}
  \input{pseudocode/synth}
  \end{multicols}
  \normalsize
  \vspace{0.5em}
\end{algorithm}

The termination of \code{Analyze} follows analogously to \code{ComputeBoolReach}.
First, note that \code{Analyze} terminates if all work items have been processed.
Each iteration of the loop at \auxlineref{Line:BoolSynth:Loop} processes at least one work item.
A state is added to the work list only if it has not yet been visited.
The number of states is finite, since Boolean programs are finite-state.
Therefore, \code{Analyze} must terminate with time polynomial in the number of program states.
This is in contrast to general synthesis, which requires time exponential in the number of program states~\cite{Vardi2008}.

\begin{theorem}
  \label{Thm:BoolComplexity}
  Let $\cP = \left( \Locs, \GV, \LV, E \right)$ with $E = \left( \Bops, \Bcalls, \Bassume, \Bassert, \Bsum \right)$ be a Boolean program.
  Then for each input $( \cP, \Pi_0 )$, \code{Analyze} of \cref{Alg:BoolSynth} terminates in $O( n^2 m \cdot |\Locs| )$ symbolic Boolean operations where $n = 2^{|\GV \cup \LV|}$ is the number of variable assignments and $m = \cdot |\Bops \cup \Bcalls \cup \Bassume \cup \Bassert|$ is the number of edges.
\end{theorem}

The correctness of \code{BoolSynth} follows from the correctness of \code{ComputeBoolReach} in~\cite{ChakiGurfinkel2018}.
\cref{Thm:Analyze} proves that \code{Analyze} extends \code{ComputeBoolReach} to obtain a least partial program summary.
\cref{Cor:BoolCorrectness} proves that an \ipsmp solution (or a witness to unrealizability) can be extracted from a least partial program summary.
Since \code{Analyze} terminates, this is a decision procedure for the Boolean case of \ipsmp.

\begin{theorem}
  \label{Thm:Analyze}
  Let $\cP = \left( \Locs, \GV, \LV, E \right)$ be a Boolean program and $\Pi_0$ be a collection of predicate templates for $\cP$.
  \code{Analyze} of \cref{Alg:BoolSynth} computes a least partial program summary, $( \theta, \sigma, \Pi )$, for $\cP$ such that $\forall p \in \Partial( \cP ) \cdot \Pi_0( p ) \implies \Pi( p )$.
\end{theorem}

\begin{corollary}
  \label{Cor:BoolCorrectness}
  \code{BoolSynth} of \cref{Alg:BoolSynth} decides \ipsmp for Boolean programs.
\end{corollary}

\subsection{The General Case}
\label{Sect:Decidability:General}

\newcommand{\SynthToCHC}{\textsf{CHCSynth}}
This section presents sound proof-rules for the realizability and unrealizability of \ipsmp problems.
These rules are shown to be instances of CHC-solving.
To justify the reduction from \ipsmp to this undecidable problem, the general case of \ipsmp is also shown to be undecidable.
First, assume that $( \cP, \cT, \Pi_0 )$ is an \ipsmp problem.
Recall that $\cP \in \Programs( \mathcal{F}, \mathcal{V} )$ where $\mathcal{F}$ is the FO-fragment of pure program expressions.
A logical encoding of $( \cP, \cT, \Pi_0 )$ is given by:
{\par\scriptsize\begin{equation*}
  \SynthToCHC( \cP, \Pi_0 ) :=
    \ToCHC( \cP ) \land \left(
      \bigwedge_{p \in \Partial( \cP )} \forall \vec{x} \cdot \left( \Pi_0( p ) \implies p( \vec{x} ) \right)
    \right)
\end{equation*}}%
The term $\ToCHC( \cP )$ encodes verification conditions for $\cP$, in which each partial predicate is unspecified.
Calls to a partial predicate $p$, under assume and assert, provide constraints on the strongest and weakest possible solutions to $\SynthToCHC( \cP, \Pi_0 )$.
The clause $\forall \vec{x} \cdot \left( \Pi_0( p ) \implies p( \vec{x} ) \right)$ then ensures the strongest solution to $p$ subsumes $\Pi_0( p )$.
Then a solution $\sigma$ to $\SynthToCHC( \cP, \Pi )$ contains an implementation $\sigma( p )$ for each partial predicate $p$, that subsumes $\Pi_0( p )$ and ensures the correctness of $\cP$ (\cref{Thm:Realizability}).
Furthermore, if $\sigma$ is an $\mathcal{F}$-solution, then each  $\sigma( p )$ can be implemented in the programming language.
On the other hand, if $\SynthToCHC( \cP, \Pi_0 )$ is unsatisfiable, then for every choice of implementation $\Pi$ satisfying $\Pi_0$, the closed program $\cP[\Pi]$ is incorrect (\cref{Thm:Unrealizability}).
Together, these theorems give sound proof rules for the realizability and unrealizability of $( \cP, \cT, \Pi_0 )$.
In practice, $\mathcal{F}$ is chosen to be the same fragment used by the CHC-solver.

\newcommand{\chcbound}{\textsf{CHCBnd}}
\begin{theorem}
  \label{Thm:Realizability}
  Let $\Sigma$ be a first-order signature, $\cV$ be a set of variable symbols, $\cF = \QFFml( \Sigma, \cV )$, $\cP \in \Programs( \Sigma, \cV )$, and $( \cP, \cT, \Pi_0 )$ be an \ipsmp problem.
  If $\sigma$ is an $\cF$-solution to $\SynthToCHC( \cP, \Pi_0 )$ relative to $\cT$, then $\Pi: \Partial( \cP ) \rightarrow \cF$ such that $\Pi: p \mapsto \sigma( p )$ is a solution to $( \cP, \cT, \Pi_0 )$.
\end{theorem}

\begin{theorem}
  \label{Thm:Unrealizability}
  If $( \cP, \cT, \Pi_0 )$ is an \ipsmp problem and $\SynthToCHC( \cP, \Pi_0 )$ is $\cT$-unsatisfiable, then $( \cP, \cT, \Pi_0 )$ is unrealizable.
\end{theorem}

$\SynthToCHC( \cP, \Pi_0 )$ strengthens $\ToCHC( \cP )$ by adding additional CHCs.
Since $\ToCHC( \cP )$ is a conjunction of CHCs, then $\SynthToCHC( \cP, \Pi_0 )$ is also a conjunction of CHCs.
Therefore, a CHC solver can check the satisfiability and unsatisfiability of $\SynthToCHC( \cP, \Pi_0 )$.
As a result, a CHC solver can find a solution to $( \cP, \cT, \Pi_0 )$ (\cref{Thm:Realizability}), or prove that the problem is unrealizable (\cref{Thm:Unrealizability}).

\begin{theorem}
  \label{Thm:SynthCHC}
  $\SynthToCHC( \cP, \Pi_0 )$ is a CHC conjunction.
\end{theorem}

\begin{example}
  This example uses \cref{Thm:Realizability} to solve the \ipsmp problem in \cref{Fig:Overview:ClassSoln}.
  The program in \cref{Fig:Overview:ClassSoln} corresponds to the \ipsmp problem $( \cP, \cT, \Pi_0 )$ where $\cP$ is the source code, $\cT$ is the theory of integer linear arithmetic, and $\Pi_0: \code{CInv} \to \bot$.
  In this example we let $\mathcal{F}$ be the fragment of linear inequalities of the variables $\{ \code{m}, \code{p} \}$, where \code{m} and \code{p} are the arguments to \code{CInv}.
  Then our goal is to find an expression $e \in \mathcal{F}$ such that $\cP[\code{CInv} \gets e]$ is correct.
  According to \cref{Thm:Realizability}, we can extract $e$ from the output of a CHC-solver.
  The first step in this process is to construct the input $\SynthToCHC( \mathcal{P}, \Pi_0 )$.
  To construct $\SynthToCHC( \mathcal{P}, \Pi_0 )$ we must first construct the term $\ToCHC( \cP )$.
  Recall that $\ToCHC( \cP )$ encodes verification conditions for the program $\cP$.
  Since $\cP$ is open (\code{CInv} is unimplemented), then \code{CInv} will be an unknown in $\ToCHC( \cP )$.
  According to \cref{Sect:Background:WLP}, $\ToCHC( \cP )$ will consist of the verification conditions for $\cP[\code{main}]$, along with a summary for each function in $\cP$.
  We begin by constructing a summary for each method from the \code{Counter} object in $\cP$.
  As described in \cref{Sect:Background:WLP}, each predicate $f_{\mathit{pre}}(\mathbf{x})$ collects the inputs $\mathbf{x}$ to a function $f$, and each predicate $f_{\mathit{sum}}(\mathbf{x}, \mathbf{e})$ each argument $\mathbf{x}$ to a return value $\textbf{e}$
  For simplicity, we encode object state by passing member fields as arguments and return values.
  Redundant declarations are omitted.
  {\par\scriptsize\begin{align*}
    \varphi_{\mathit{Ctor}} :=\;& \forall m \cdot \mathit{Counter}_{\mathit{pre}}(m) \implies \left( (m > 0) \implies \mathit{Counter}_{\mathit{sum}}( m, m, 0 ) \right) \\
    \varphi_{\mathit{Reset}} :=\;& \forall m \cdot \forall p \cdot \mathit{reset}_{\mathit{pre}}( m, p ) \implies \mathit{reset}_{\mathit{post}}( m, p, m, 0 ) \\
    \varphi_{\mathit{Cap}} :=\;& \forall m \cdot \forall p \cdot \mathit{capacity}_{\mathit{pre}}( m, p ) \implies \mathit{capacity}_{\mathit{sum}}( m, p. m - p \ne 0 ) \\
    \varphi_{\mathit{Incr}} :=\;& \forall m \cdot \forall p \cdot \mathit{increment}_{\mathit{pre}}( m, p ) \implies (
    \begin{aligned}[t]
        & \left(( p \ge m ) \implies \mathit{increment}_{\mathit{sum}}( m, p, \bot ) \right) \land \\
        & \left(( p < m ) \implies \mathit{increment}_{\mathit{sum}}( m, p + 1, \top ) \right))
    \end{aligned}
  \end{align*}}%
  Next, we construct a summary for the function \code{drain}.
  Note that, unlike the methods of \code{Counter}, the function \code{drain} contains a loop.
  As described in \cref{Sect:Background:WLP}, loops are encoded using loop invariants with the loop at line $n$ associated with an invariant $\mathit{loop}_n$.
  In our example, the loop at \cref{Line:Overview:ClassProb:Loop} of \cref{Fig:Overview:ClassProb} is associated with a loop invariant $\mathit{loop}_{13}$.
  Then the summary of \code{drain} is as follows.
  {\par\scriptsize\begin{align*}
    \varphi_{\mathit{Exit}} :=&\; \mathit{capacity}_{\mathit{pre}}( p', m' ) \cdot \forall x \cdot \left( \mathit{capacity}_{\mathit{sum}}( p', m', x) \implies \left( x \land \mathit{drain}_{\mathit{sum}}( p, m, p', m' ) \right) \right) \\
    \varphi_{\mathit{Loop}} :=&\; \mathit{loop}_{13}( p, m, x ) \land \\
    &\; ((\mathit{loop}_{13}( p, m, x ) \land x > 0) \implies (\mathit{capacity}_{\mathit{pre}}( p, m ) \land \forall x \cdot ( \mathit{capacity}_{\mathit{sum}}( p, m x ) \implies \mathit{loop}_{13}))) \land \\
    &\; ((\mathit{loop}_{13}( p, m, x ) \land x \le 0) \implies \mathit{reset}_{\mathit{pre}}( p, m ) \land \forall p' \cdot \forall m' \cdot \left( \mathit{reset}_{\mathit{sum}}( p, m, p', m' ) \implies \varphi_{\mathit{Exit}} \right)) \\
    \varphi_{\mathit{Dr}} :=& \forall p \cdot \forall m \cdot \mathit{drain}_{\mathit{pre}}(p, m) \implies 
    ( \mathit{capacity}_{\mathit{pre}}( p, m ) \land \forall x \cdot ( \mathit{capacity}_{\mathit{sum}}( p, m, x ) \implies \varphi_{Loop} ))
  \end{align*}}%
  Finally, we construct the verification conditions for \code{main}.
  Since \code{main} is the entry-point to $\cP$, then \code{main} must be safe for all possible inputs.
  This means that \code{main} does not require a summary.
  The conditions are as follows.
  {\par\scriptsize\begin{align*}
    \varphi_{\mathit{Main}} :=\;& \forall b_1 \cdot \forall b_2 \cdot \forall b_3 \cdot \\
    & \begin{aligned}[t]
        & (b_1 = 1) \implies \left( \forall m \cdot \mathit{Counter}_{\mathit{pre}}( m ) \land \forall m' \cdot \forall p \cdot ( \mathit{Counter}_{\mathit{sum}}( m, m', p ) \implies \mathit{CInv}( m', p ) \right) \land \\
        & (b_1 \ne 1 \land b_2 = 1) \implies ( \forall m \cdot \forall p \cdot \mathit{CInv}( m, p ) \implies \\
        &\qquad \left( \mathit{reset}_{\mathit{pre}}( m, p ) \land \forall m' \cdot \forall p' \left( \mathit{reset}_{\mathit{sum}}( p, m, p', m' ) \implies \mathit{CInv}( p', m' ) \right) \right)) \land \\
        & (b_1 \ne 1 \land b_2 \ne 1 \land b_3 = 1) \implies ( \forall m \cdot \forall p \cdot \mathit{CInv}( m, p ) \implies \\
        &\qquad \left( \mathit{increment}_{\mathit{pre}}( m, p ) \land \forall m' \cdot \forall p' \left( \mathit{increment}_{\mathit{sum}}( p, m, p', m' ) \implies \mathit{CInv}( p', m' ) \right) \right)) \land \\
        & (b_1 \ne 1 \land b_2 \ne 1 \ne b_3 = 1) \implies ( \forall m \cdot \forall p \cdot \mathit{CInv}( m, p ) \implies \\
        &\qquad \left( \mathit{drain}_{\mathit{pre}}( m, p ) \land \forall m' \cdot \forall p' \left( \mathit{drain}_{\mathit{sum}}( p, m, p', m' ) \implies \top \right) \right))
    \end{aligned}
  \end{align*}}%
  As outlined in \cref{Sect:Background:WLP}, $\ToCHC( \mathcal{P} ) = \varphi_{\mathit{Main}} \land \varphi_{\mathit{Ctor}} \land \varphi_{\mathit{Reset}} \land \varphi_{\mathit{Cap}} \land \varphi_{\mathit{Incr}} \land \varphi_{\mathit{Dr}}$.
  Next, $\ToCHC( \cP )$ is strengthened by the predicate template $\Pi_0( \code{CInv} )$ to obtain $\SynthToCHC( \cP, \Pi_0 ) = \ToCHC( \cP ) \land \left(
    \forall m \cdot \forall p \cdot \bot \implies \code{CInv}(m, p) \right)$.
  Clearly the term $\bot \implies \code{CInv}(m, p)$ is trivially satisfied.
  This is because the predicate template $\Pi_0( \code{CInv} )$ is also trivial.
  In general, this need not be the case.
  Nonetheless, the term $\ToCHC( \cP )$ is non-trivial.
  If $\ToCHC( \cP )$ is provided to a CHC-solver, then the CHC-solver will return a solution $\sigma$ containing the following components:
  expressions $\sigma( \mathit{Counter}_{\mathit{pre}} )$, $\sigma( \mathit{Reset}_{\mathit{pre}} )$, $\sigma( \mathit{Capacity}_{\mathit{pre}} )$, $\sigma( \mathit{Increment}_{\mathit{pre}} )$, and $\sigma( \mathit{Drain}_{\mathit{pre}} )$, which over-approximate the inputs passed to each function;
  expressions $\sigma( \mathit{Counter}_{\mathit{sum}} )$, $\sigma( \mathit{Reset}_{\mathit{sum}} )$, $\sigma( \mathit{Capacity}_{\mathit{sum}} )$, $\sigma( \mathit{Increment}_{\mathit{sum}} )$, and $\sigma( \mathit{Drain}_{\mathit{sum}} )$, which over-approximate the return values of each function;
  an expression $\sigma( \mathit{loop}_{13} )$ which over-approximates the reachable states of the loop in \code{drain};
  an expression $\sigma( \code{CInv} )$ which describes a safe implementation for \code{CInv}.
  In one solution, $\sigma( \mathit{loop}_{13} ) = (p \le m \land (x \ne 0 \implies 0 < p) \land (x = 0 \implies 0 = p))$.
  This states that the counter is always in a valid position, and in position zero if and only if the capacity returns to zero.
  In such a solution, it is also possible that $\sigma( \code{CInv} ) = (m > 0 \land p \le m )$.
  Clearly $\sigma( \code{CInv} )$ is an $\mathcal{F}$-solution since $\sigma( \code{CInv} )$ is a conjunction of linear inequalities.
  Then by \cref{Thm:Realizability}, $\Pi: \code{CInv} \to (m > 0 \land p \le m )$ is a solution to $( \mathcal{P}, \mathcal{T}, \Pi_0 )$ with $\cP[\Pi]$ both closed and safe.
  \qed
\end{example}

Like CHC-solving, the general \ipsmp problem is also undecidable.
This is because program verification reduces to \ipsmp.
Intuitively, if a closed program $\cP$ is given to an \ipsmp solver, then a solution to the \ipsmp problem implies that $\cP$ is correct, and a witness to unrealizability implies that $\cP$ is incorrect.
In this way, the halting problem also reduces to \ipsmp.

We show that \ipsmp is undecidable for linear integer arithmetic by reducing the halting problem for $2$-counter machines to \ipsmp.
Recall that a $2$-counter machine is a program with a program counter and two integer variables~\cite{Minsky1967}.
The program has a finite number of locations, each with one of four instructions:
(1) \code{inc(x)} increases the variable \code{x} by $1$ and increment the program counter;
(2) \code{dec(x)} decreases the variable \code{x} by $1$ and increment the program counter;
(3) \code{jump(x, i)} goes to location \code{i} if \code{x} is $0$, else increments the program counter; 
(4) \code{halt()} halts execution of the program.
The halting program for $2$-counter machines is known to be undecidable~\cite{Minsky1967}.

\begin{theorem}
  \label{Thm:Undecidable}
  The \ipsmp problem is undecidable for linear integer arithmetic.
\end{theorem}

%% file: pseudocode/synth.tex
\SetKwFunction{FInit}{Init}
\SetKwFunction{FDoIntraprocedural}{DoIntraproc}
\SetKwFunction{FDoProcedures}{DoProcs}
\SetKwFunction{FDoAssumptions}{DoAssumes}
\SetKwFunction{FDoAssertions}{DoAsserts}
\SetKwFunction{FDoFunctionSummary}{DoFuncSum}
\SetKwFunction{FDoProcedureSummary}{DoProcSum}
\SetKwFunction{FAnalyze}{Analyze}
\SetKwFunction{FBoolSynth}{BoolSynth}
\SetKwProg{Fn}{Func}{:}{}

$\var\;( \theta, \sigma, \Pi )$
\tcp*{A partial program summary}
$\var\; \worklist$
\tcp*{A map from $\Locs$ to queued states}

\Fn{\FDoAssumptions{$V$, $\LV$, $\Bsum$, $\Bassume$, $\lwork$, $\vwork$}}{
  \For{$( \lwork, \lin, \lret ) \in \Bassume$}{
    $\vin \gets \bddexists( \vwork, V \cup \LV' )[ V' / V ]$\;
    \FUpdateReach{$\lret$, $\Pi( \lin ) \land \vin$}\;
  }
}

\Fn{\FDoAssertions{$V$, $\LV$, $\Bsum$, $\Bassert$, $\lwork$, $\vwork$}}{
  \For{$( \lwork, \lin, \lret ) \in \Bassert$}{
    \FUpdateReach{$\lret$, $\Pi( \lin ) \land \vwork$} \;
    \FUpdateReach{$\lin$, $\bddexists( \vwork, V \cup \LV' )[ V' / V ]$}\;
  }
}

\Fn{\FDoFunctionSummary{$V$, $\LV$, $\Bassume$, $\Bassert$, $\lwork$, $\vwork$}}{
  \uIf{$\lwork \in \Partial( \cP )$}{
    $\Pi( \lwork ) \gets \Pi( \lwork ) \lor \vwork$\;
    \For{$(\lcall, \lwork, \lret ) \in \Bassume \cup \Bassert$}{
      \FUpdateReach{$\lret$, $\theta( \lcall ) \land \vwork$}
   }
  }
}

\Fn{\FInit{$\Locs$, $\Bsum$, $\Pi_0$}}{
  \FInitBoolReach{$\Locs$, $\Bsum$}\;
  \lFor{$l \in \Partial( \cP )$}{
    $\Pi( l ) \gets \Pi_0( l )$\label{Line:BoolSynth:InitPi}
  }
}

\Fn{\FAnalyze{$\cP$, $\Pi_0$}}{
  $( \Locs, \GV, \LV, ( N, C, F, A, P ) ) \gets \cP$\;
  $V \gets \GV \cup \LV$ ;
  \FInit{$\Locs$, P, $\Pi_0$}\;
  \label{Line:BoolSynth:InitAll}
  \While{$\exists \;\lwork \in \Locs \cdot \worklist( \lwork ) \ne \bot$}{ \label{Line:BoolSynth:Loop}
    $\vwork \gets \worklist( \lwork )$ ;
    $\worklist( \lwork ) \gets \bot$\;
    \FDoIntraprocedural{$V$, N, $\lwork$, $\vwork$}\;
    \FDoProcedures{$V$, $\LV$, P, C, $\lwork$, $\vwork$}\;
    \FDoAssumptions{$V$, $\LV$, P, F, $\lwork$, $\vwork$}\;
    \FDoAssertions{$V$, $\LV$, P, A, $\lwork$, $\vwork$}\;
    \FDoProcedureSummary{$V$, $\LV$, P, C, $\lwork$, $\vwork$}\;
    \FDoFunctionSummary{$V$, $\LV$, F, A, $\lwork$, $\vwork$};
  }
  \label{Line:BoolSynth:FixedPoint}
}

\Fn{\FBoolSynth{$\cP$, $\Pi_0$}}{
  \FAnalyze{$\cP$, $\Pi_0$}\;
  \label{Line:BoolSynth:CallAnalyze}
  \lIf{$\theta( \lbot ) = \bot$}{
    return $( \checkmark, \Pi )$
  }
  \label{Line:BoolSynth:Decide}
  \lElse{
    return $( \times, \Pi )$
  }
}

%% file: reductions.tex
\section{From Verification to Synthesis}
\label{Sect:Reductions}

This section establishes reductions of \cref{Sect:Overview}.
Class invariant inference is proven directly.
Array abstraction and symmetric ring verification are subsumed by a reduction from parameterized compositional model checking to \ipsmp.
Loop invariant synthesis is proven in \appendixcite{Appendix:LoopReduction}.
We write $\Sigma$ for a first-order signature, $\mathcal{V}$ for a set of variables, and $\Pi_{\bot}$ for a collection of predicate templates which maps each predicate to $\bot$.

\subsection{Class Invariant Inference}
\label{Sect:Reductions:Class}

\begin{figure}[t]
  \begin{subfigure}{0.49\textwidth}
\begin{lstlisting}[style=ipsmp]
class Cls {
  int x; int y;
  Cls(int a) { ... }
  void f(int a) { ... }
  void g(int a, int b) { ... }}
void func(Cls ob, int a) { ... }
\end{lstlisting}
    \caption{The input program.}
    \label{Fig:Reduction:ClassProb}
  \end{subfigure}
  \hfill
  \begin{subfigure}{0.49\textwidth}
\begin{lstlisting}[style=ipsmp]
°bool PRED_TEMPLATE Inv(int x, int y) {°
  °return synth(x, y); }°
void main(int br, int a, int b) {
  if (br == 0) {
    Cls ob = Cls(a);
    ~assert(Inv(ob));~
  } else if (br == 1) {
    ^Cls ob = *; assume(Inv(ob));^
    ob.f(a);
    ~assert(Inv(ob));~
  } else if (br==2) {
    ^Cls ob = *; assume(Inv(ob));^
    ob.g(a, b);
    ~assert(Inv(ob));~
  } else if (br==3) {
    ^Cls ob = *; assume(Inv(ob));^
    func(ob, a); }}
\end{lstlisting}
    \caption{The \ipsmp reduction.}
    \label{Fig:Reduction:ClassSoln}
  \end{subfigure}
  \caption{A reduction from class invariant inference to \ipsmp.}
\end{figure}

A \emph{safe class invariant} is a predicate that is true of a class instance after initialization, closed under the execution of each impure class method, and sufficient to prove the correctness of a function taking class instances as arguments~\cite{HuizingKuiper2000}.
\emph{Class invariant inference} asks to find a safe class invariant given a program.
The inference problem is \emph{intensional} if solutions are in the same logical fragment as assertions in the programming language~\cite{NielsonNielson2007}.
A definition of (intensional) class invariant inference is found in \cref{Def:ClassInvInference}.
In this definition, $\ToCHC( f )$ relates the class invariant $\varphi$ to a summary of each method $f$ in $\cP$, and $f_{pre}$ is used to enforce that $f$ is summarized.
For simplicity, a class has two fields and two impure methods, each taking at most two arguments (\cref{Fig:Reduction:ClassProb}).
A generalization to $m$ methods is not difficult.
A generalization to $n$ arguments follows immediately.

\begin{definition}
  \label{Def:ClassInvInference}
  A \emph{class invariant inference problem} is a tuple $( \cP, \cT )$ such that $\cP \in \Programs( \Sigma, \cV )$ is an open program as in \cref{Fig:Reduction:ClassProb} and $\cT$ is a theory.
  A solution to $( \cP, \cT )$ is a $\varphi \in \QFFml( \Sigma, \{ \code{x}, \code{y} \} )$ such that the following are $\cT$-satisfiable:
  {\par\scriptsize\begin{align*}
    \psi_{\mathit{Init}} &:= \forall V \left( \mmsmallcode{Cls}_{pre}( a ) \land \mmsmallcode{Cls}_{sum}( a, x, y ) \implies \varphi \right)
    &
    \psi_{\mathit{Close1}} &:= \forall V \left( \varphi \land \mmsmallcode{f}_{sum}( x, y, a, x', y' ) \implies \varphi' \right) \\
    \psi_{\mathit{Close2}} &:= \forall V \left( \varphi \land \mmsmallcode{g}_{sum}( x, y, a, b, x', y' ) \implies \varphi' \right)
    &
    \psi_{\mathit{Suffic}} &:= \forall V \left( \varphi \implies \mmsmallcode{func}_{sum}( x, y, a ) \right)
  \end{align*}
  \vspace{-\belowdisplayskip}
  \vspace{-\abovedisplayskip}
  \begin{align*}
    \psi_{\mathit{Sum}} &:= \ToCHC( \cP ) \land \forall V ( \varphi \implies \mmsmallcode{f}_{pre}( x, y, a ) \land \mmsmallcode{g}_{pre}( x, y, a, b ) \land \mmsmallcode{func}_{pre}( x, y, a ) )
  \end{align*}}%
\end{definition}

\begin{theorem}
  \label{Thm:ClassReduction}
  Let $( \cP, \cT )$ be a class invariant inference problem and $\cP'$ be the
  program obtained by adding \code{main} in
  \cref{Fig:Reduction:ClassSoln} to $\cP$. Then $\Pi$ is a solution to $( \cP', \Pi_{\bot}, \cT )$ if and only if $\Pi(
  \code{Inv} )$ is a solution to $( \cP, \cT )$.
\end{theorem}

\subsection{Reducing PCMC to \ipsmp}
\label{Sect:Reductions:PCMC}

\begin{figure}[t]
  \begin{subfigure}{0.49\textwidth}
\begin{lstlisting}[style=ipsmp]
struct View { int l; int r; int s; };
// Transition relation.
View tr(View v) { ... }
// Initial state predicate.
bool init(int l, int s, int r) { ... }
// Correctness property.
bool property(View v) { ... }
\end{lstlisting}
    \caption{The input program.}
    \label{Fig:Reduction:PCMSProb}
  \end{subfigure}
  \hfill
  \begin{subfigure}{0.49\textwidth}
\begin{lstlisting}[style=ipsmp]
°bool PRED_TEMPLATE Inv(°
  °int l, int s, int, r) {°
  °if (init(l, s, r)) { return true; }°
  °else { return synth(l, s, r); } }°
void main(int br, struct View v) {
  if (br == 0) {
    ^int inf = *;^
    ^assume(Inv(v.left, v.st, v.right));^
    ^assume(Inv(v.right, inf, v.left));^
    v = tr(v);
    ~assert(Inv(v.left, v.st, v.right));~
    ~assert(Inv(v.right, inf, v.left));~
  } else if (br == 1) {
    ^assume(Inv(v.left, v.st, v.right));^
    assert(property(v)); }}
\end{lstlisting}
      \caption{The \ipsmp reduction.}
      \label{Fig:Reduction:PCMCSoln}
  \end{subfigure}
  \caption{A reduction from compositional ring invariant synthesis to \ipsmp. The state of each process and resource are both assumed to be integer values.}
\end{figure}

\emph{Parameterized compositional model checking}~(PCMC) is a framework to verify structures with arbitrarily many components~(e.g., an array with arbitrarily many cells, or a ring with arbitrarily many processes) by decomposing the structure into smaller structures of fixed sizes~\cite{NamjoshiTrefler2016}.
Intuitively, each of these smaller structures is a view of the larger structure from the perspective of a single component.
A proof of the larger structure is obtained by verifying each of the smaller structures, and showing that their proofs compose with one another~\cite{NamjoshiTrefler2016}.
If the number of smaller structures is finite (i.e.,~most perspectives are similar), then PCMC is applicable~\cite{NamjoshiTrefler2016}.
For example, in \cref{Sect:Overview:Array} and \cref{Sect:Overview:PCMC}, the array and ring were highly symmetric, and therefore, all perspectives were similar.

Once the larger structure has been decomposed, the proof of compositionality follows by inferring \emph{adequate compositional invariants} for groups of similar components~\cite{NamjoshiTrefler2016}.
The number of compositional invariants, and the properties they must satisfy, depend on the decomposition.
However, each property is one of initialization, closure, or non-interference.
An initialization property states that a compositional invariant is true for the initial state of a component.
A closure property states that a compositional invariant is closed under all transitions of its components.
A non-interference property states that for any component $\mathit{c}$, if $\mathit{c}$ satisfies its compositional invariant and an adjacent component (\emph{also satisfying its compositional invariant}) performs a transition, then $\mathit{c}$ continues to satisfy its compositional invariant after the transition.
In addition, all composition invariants must be adequate in that they imply the correctness of the larger structure.
To make the rest of this section concrete, we restrict ourselves to compositional ring invariants\footnote{\cref{Sect:Overview:Array} is a degenerate case. In this ring, processes communicate through locks. In an array, cells do not ``communicate''.}.
As in \cref{Sect:Reductions:Class}, the inference problem is assumed to be intensional.
A formal definition of (intensional) compositional invariant inference is given in \cref{Def:CompInvInference}\footnote{In PCMC, a witness to unrealizability does not entail the incorrectness of a structure. Instead, no proof of correctness exists relative to the chosen decomposition.}.
Note that in \cref{Def:CompInvInference} $\ToCHC$ relates the compositional invariant $\varphi$ to the summary of \code{tr}, $\code{tr}_{pre}$ enforces that \code{tr} is summarized, and $\varphi_{\mathit{Inf}} := \varphi[ l / r ][ s / i ][ r / l ]$ is the compositional invariant applied to a process $( r, i, l )$.

\begin{definition}
  \label{Def:CompInvInference}
  A \emph{compositional ring invariant (CRI) inference problem} is a tuple $( \cP, \cT )$ such that $\cP \in \Programs( \Sigma, \cV )$ is an open program as in \cref{Fig:Reduction:PCMSProb} and $\cT$ is a theory.
  A solution to $( \cP, \cT )$ is a $\varphi \in \QFFml( \Sigma, \{ \code{l}, \code{s}, \code{r} \} )$ such that the following are $\cT$-satisfiable given $\varphi_{\mathit{Inf}} := \varphi[ l / r ][ s / i ][ r / l ]$:
  {\scriptsize\begin{align*}
    \psi_{\mathit{Init}} &:= \forall V \left( \mmsmallcode{init}( l, s, r ) \implies \varphi \right)
    &
    \psi_{\mathit{Close}} &:= \forall V \left( \varphi \land \varphi_{\mathit{Inf}} \land \mmsmallcode{tr}_{sum}( l, s, r, l', s', r' ) \implies \varphi' \right)
    \\
    \psi_{\mathit{Adeq}} &:= \forall V \left( \varphi \implies \mmsmallcode{property}( l, s, r ) \right)
    &
    \psi_{\mathit{Inf}} &:= \forall V \left( \varphi \land \varphi_{\mathit{Inf}} \land \mmsmallcode{tr}_{sum}( l, s, r, l', s', r' ) \implies {\varphi_{\mathit{Inf}}}' \right)
  \end{align*}
  \vspace{-\belowdisplayskip}
  \vspace{-\abovedisplayskip}
  \begin{align*}
    \psi_{\mathit{Sum}} &:= \ToCHC( \cP ) \land \forall V \cdot \left( \varphi \land \varphi_{\mathit{Inf}} \implies \mmsmallcode{tr}_{pre}( l, s, r ) \right)
  \end{align*}}%
\end{definition}

\begin{theorem}
  \label{Thm:ParamReduction}
  Let $( \cP, \cT )$ be a CRI inference problem, $\cP'$ be the program obtained by adding \code{main} of \cref{Fig:Reduction:PCMCSoln} to $\cP$, and $\Pi_0$ be the predicate template from \cref{Fig:Reduction:PCMCSoln}.
  Then $\Pi$ is a solution to $( \cP', \Pi_0, \cT )$ if and only if $\Pi( \code{Inv} )$ is a solution to $( \cP, \cT )$.
\end{theorem}

%% file: implementation.tex
\section{Implementation and Evaluation}
\label{Sect:Implementation}

We have implemented an \ipsmp solver within the \seahorn verification framework.
\seahorn takes as input a C program, and returns a CHC-based verification problem in the SMT-LIB format according to~\cref{Sect:Background:WLP}~\cite{GurfinkelKahsai2015}.
We extend \seahorn to recognize predicate templates.
For each predicate, \seahorn adds clauses to the verification conditions according to \cref{Sect:Decidability:General}.
Proofs of unrealizability are generated with the implementation of ~\cite{10.1007/978-3-030-03592-1_2} found in \seahorn.
That is, proofs of unrealizability are already supported by \seahorn.

The goal of our evaluation is to confirm that:
\begin{enumerate}
    \item \ipsmp is practical for the reduction described in \cref{Sect:Reductions};
    \item CHC-based solvers are more efficient than general synthesis solvers for \ipsmp instances;
    \item The overhead incurred when using \ipsmp is tolerable.
\end{enumerate}
Towards (1) and (2), we have collected 92 \ipsmp problems with linear integer arithmetic as the background theory (see~\textsf{Safe} in~\cref{Table:Benchmarks}).
Of these benchmarks, 7 reflect loop invariant inference (and interpolation~\cite{McMillan2003}), 6 reflect class invariant synthesis, 4 reflect array (and memory) abstraction, 2 reflect ring PCMC, 3 reflect procedure summarization, and 70 reflect parameterized analysis of smart-contract (SC) programs~(see~\cite{WesleyEtAl2021,WesleyEtAl2022}).
The first $20$ benchmarks were collected from research papers in the area of software verification.
The remaining benchmarks, involving the parameterized analysis of SCs, were obtained by extending \smartace with support for \ipsmp.
Of these 70 SC benchmarks, 62 are taken from real-world examples used to manage monetary assets~\cite{PermenevDimitrov2020}.
To address question (3), we compare the performance of \smartace with and without \ipsmp, relative to these real-world examples.
Note that the extension to \smartace was a routine exercise, due to the original design of \smartace.
In particular, \smartace encodes all compositional invariants as predicates returning \code{true}, to then be refined manually by an end-user~\cite{WesleyEtAl2022}.
These predicates appear in assume and assert statements, as described in~\cref{Sect:Overview:PCMC}.
Our extended version of \smartace can replace these predicates with predicate templates, yielding valid \ipsmp problems.

\begin{table*}[t]
  \small
  \centering
  \scalebox{0.65}{\input{diagrams/benchmarks}}
  \vspace{5pt}
  \caption{Performance of various solvers on \ipsmp benchmarks.}
  \label{Table:Benchmarks}
\end{table*}

A summary of all benchmarks can be found in \cref{Table:Benchmarks}.
As reflected by their size (see~\textsf{Size} in~\cref{Table:Benchmarks}) and total number of unknown predicates across all realizable instances (see~\textsf{Preds} in~\cref{Table:Benchmarks}), SCs are included to evaluate \ipsmp on large programs.
When possible, benchmarks are drawn from prior works in program analysis~(i.e.,~\cite{KahsaiKersten2017,Logozzo2004,Schwerhoff2016,PermenevDimitrov2020}).
To reflect unrealizability in \ipsmp, 29 faults have been injected in these benchmarks (see~\textsf{Buggy} in~\cref{Table:Benchmarks}).
Further information can be found about the realizable real-world SC's in \cref{Table:Overhead}.
Each SC in this table is associated with one or more safety properties (see~\textsf{Props} in~\cref{Table:Overhead}), which in turn, corresponds to a realizable \ipsmp instance.
As before, \textsf{Preds} and \textsf{Size} indicate the total predicate count and size for these instances.
All benchmarks are available at \url{https://doi.org/10.5281/zenodo.5083785}.

\newcommand{\trademark}{\textsuperscript{\textregistered}}
To evaluate \ipsmp, we find the number of benchmarks that are solved by either of two state-of-the-art CHC solvers: \eldarica~\cite{HojjatRummer2018} and \spacer~\cite{KomuravelliGurfinkel2014}.
To compare CHC solvers to general synthesis tools, we provide our benchmarks to a state-of-the-art specification synthesizer, \hornspec~\cite{PrabhuFedyukovich2021}, and a state-of-the-art \sygus solver, \cvc\footnote{To support \cvc, we convert each \emph{realizable} problem from SMT-LIB format to the \sygus input language.}~\cite{BarretConway2011}.
Since \cvc solves \sygus instances, which do not support proofs on unrealizability, then we only evaluate \cvc on realizable benchmarks (see~\textsf{N/A} in~\cref{Table:Benchmarks}).
Due to the size of each SC benchmark, we only ran the tools that could solve \textsf{Loop} through to \textsf{Proc} on these benchmarks.
The results for each tool are reported in \cref{Table:Benchmarks}, where \textsf{TO} is the number of timeouts (after 30 minutes), \textsf{MEM} is the number of failures due to memory limits, \textsf{UN} is the number of benchmarks for which a tool returned \code{unknown}, $\checkmark$ is the number of benchmarks solved, and \textsf{Time} is the total time (in seconds) to find all solutions in a given set.
In \cref{Table:Overhead}, the total time for \spacer is further broken down by SC (see~\smartace(\ipsmp) in~\cref{Table:Overhead}).
For comparison, the verification times for \verx (an automated SC verifier with user-guided predicate abstraction~\cite{PermenevDimitrov2020}) and the original version of \smartace (see~\smartace(Manual) in~\cref{Table:Overhead}) are also provided.
All evaluations were run on an Intel\trademark~Core i7\trademark~CPU~@~1.8GHz 8-core machine with 16GB of RAM running Ubuntu~20.04.

\begin{table*}[t]
  \small
  \centering
  \scalebox{0.65}{\input{diagrams/overhead}}
  \vspace{5pt}
  \caption{Overhead of integrating \ipsmp-solving with \smartace.}
  \label{Table:Overhead}
\end{table*}

From this evaluation, we answer questions (1) through to (3) in the positive.
\begin{enumerate}
    \item As illustrated by \cref{Table:Benchmarks}, many examples of class invariant inference and compositional invariant inference (i.e.,~\textsc{Class}, \textsc{Array}, \textsc{Ring}, and \textsc{SC}) taken from the literature could be encoded using \ipsmp.
    In the case of SC, the generation of \ipsmp instances could be fully-automated using a modified version of \smartace.
    We conclude that \ipsmp is practical for the reductions described in \cref{Sect:Reductions}.
    \item  As shown in \cref{Table:Benchmarks}, all small benchmarks were solved by \eldarica and \spacer, with average times under a minute. 
    Furthermore, all but four SC benchmarks were solved by \spacer within a $30$-minute timeout, with an average time of $96$ seconds.
    Upon closer inspection, we found that \spacer would fail to solve these four examples, and would return \code{unknown} after approximately one hour.
    However, \cvc failed to solve any SC benchmarks within $30$-minutes.
    Therefore, we conclude that CHC-based \ipsmp-solving is effective for the reductions of \cref{Sect:Reductions}, and can outperform general synthesis solutions.
    We note that \hornspec returned \code{unknown} on all but two benchmarks\footnote{The authors of \hornspec confirm this result though the cause is unknown.}.
    \item As shown in \cref{Table:Overhead}, the \ipsmp version of \smartace incurred an average time overhead of 18x as compared to the manual version of \smartace.
    This should come as no surprise, since the manual version of \smartace achieved a verification time of under 3 seconds for 44 of the 62 properties with the help of user-provided compositional invariants.
    In these cases, a solving time as low as 60 seconds would correspond to an overhead of at least 20x.
    To better contextualize this overhead, we compare the verification time of \ipsmp version of \smartace to the verification time of \verx.
    We first note the outlying case of \textsc{PolicyPal}, in which the \ipsmp version of \smartace achieves a speedup of over 6x.
    For the remaining SC's, the \ipsmp version of \smartace fell within 1.3x of \verx on average.
    Since \verx is a specialized tool with less automation than the \ipsmp version of \smartace, we conclude that the overhead incurred by \ipsmp is tolerable in this particular real-world application.
    We note that in~\cite{PermenevDimitrov2020}, only the ``average'' times were reported for \verx.
    It is unclear whether this is the average time to verify all properties, or an average across all properties.
    The authors of \verx were contacted, but were unable to provide the original data.
    For this reason, we assume conservatively that all times reported by \verx are total.
\end{enumerate}
One limitation of the evaluation is its emphasis on SC verification.
However, compositional SC verification is representative of compositional verification, as illustrated in~\cite{WesleyEtAl2021}.
We do acknowledge that design patterns specific to SC development might bias the benchmark set.
We hope for this benchmark set to be expanded in future work.

Note, however, that we do not plan to benchmark our IPS-MP solver against invariant synthesis tools.
Recall that our implementation simply extends \seahorn with support for the IPS-MP synthesis language.
In cases where the IPS-MP instance reduces to invariant synthesis, our extension is bypassed, and verification reduces to executing \seahorn.
Therefore, a direct comparison is not possible, and the evaluation results would not be meaningful.
Furthermore, \seahorn is a state-of-the-art program verifier with prior success in SV-COMP.
Thus, \seahorn is already known to perform well on invariant synthesis tasks.

An important direction for future work is to understand why \cvc times out on all benchmarks.
We hypothesize that the lack of a grammar in \ipsmp proves challenging for \cvc's enumerative search.
We also note that many of our benchmarks produce non-linear CHC's, whereas the invariant synthesis track for \sygus reduces to solving linear CHC's.

%% file: diagrams/benchmarks.tex
\begin{tabular}{@{}lrrrrrrrrrrrrrrrrrrrrr@{}}
    \toprule \multicolumn{5}{c}{Benchmarks} &
    \phantom{a} & \multicolumn{3}{c}{\ipsmp(\spacer)} &
    \phantom{a} & \multicolumn{4}{c}{\ipsmp(\eldarica)} &
    \phantom{a} & \multicolumn{3}{c}{\hornspec} &
    \phantom{a} & \multicolumn{3}{c}{\cvc} \\

    \cmidrule{1-5} \cmidrule{7-9} \cmidrule{11-14} \cmidrule{16-18} \cmidrule{20-22}
    Type & Safe & Buggy & Preds & Size &&
    Time & TO & $\checkmark$ &&
    Time & TO & MEM & $\checkmark$ &&
    Time & UN & $\checkmark$ &&
    TO & N/A & $\checkmark$ \\

    \midrule
    Loop & 7 & 7 & 9 & 179~KB && \bf 4 & 0 & \bf 14 && 45 & 0 & 0 & \bf 14 && 4 & 12 & 2 && 7 & 7 & 0 \\
    Class & 6 & 6 & 6 & 694~KB && \bf 2 & 0 & \bf 12 && 1\,449 & 0 & 0 & \bf 12 && --- & 12 & 0 && 6 & 6 & 0 \\
    Array & 4 & 6 & 6 & 535~KB && \bf 4 & 0 & \bf 10 && 230 & 0 & 0 & \bf 10 && --- & 10 & 0 && 4 & 6 & 0 \\
    Ring & 2 & 3 & 2 & 197~KB && \bf 1 & 0 & \bf 5 && 52 & 0 & 0 & \bf 5 && --- & 5 & 0 && 2 & 3 & 0 \\
    Proc & 3 & 3 & 3 & 418~KB && \bf 2 & 0 & \bf 6 && 4 & 0 & 0 & \bf 6 && --- & 6 & 0 && 3 & 3 & 0 \\
    SC & 70 & 4 & 181 & 974~MB && \bf 6\,878 & 4 & \bf 70 && 6\,717 & 53 & 12 & 9 && --- & --- & --- && --- & --- & --- \\
    \midrule
    \bf Total & 92 & 29 & 207 & 975~MB && \bf 6\,891 & 4 & \bf 117 && 8\,497 & 53 & 12 & 56 && 4 & 45 & 2 && 22 & 29 & 0 \\
    \bottomrule
\end{tabular}

%% file: diagrams/overhead.tex
\begin{tabular}{@{}lrrrrrrrr@{}}
    \toprule \multicolumn{5}{c}{Contracts} &
    \phantom{abc} & \multicolumn{3}{c}{Performance (Time)} \\

    \cmidrule{1-5} \cmidrule{7-9}
    Name & Props & Preds & Size & $\checkmark$ &&
    \verx~\cite{PermenevDimitrov2020} & \smartace(Manual)~\cite{WesleyEtAl2021} & \smartace(\ipsmp) \\

    \midrule
    Alchemist & 3 & 12 & 36 MB & 3 && 29 & 7 & 208 \\
    Brickblock & 6 & 12 & 122 MB & 6 && 191 & 13 & 1214 \\
    Crowdsale & 9 & 27 & 76 MB & 9 && 261 & 223 & 238 \\
    ERC20 & 9 & 27 & 45 MB & 9 && 158 & 12 & 103 \\
    Melon & 16 & 32 & 149 MB & 16 && 408 & 30 & 979 \\
    PolicyPal & 4 & 16 & 123 MB & 4 && 20\,773 & 26 & 3118 \\
    VUToken & 5 & 22 & 319 MB & 1 && 715 & 19 & 17 \\
    Zebi & 5 & 14 & 45 MB & 5 && 77 & 8 & 487 \\
    Zilliqa & 5 & 10 & 54 MB & 5 && 94 & 8 & 501 \\
    \midrule
    \bf Total & 62 & 172 & 969 MB & 58 && 22\,706 & 346 & 5685 \\
    \bottomrule
\end{tabular}

%% file: related.tex
\section{Related Work}

\noindent
\textbf{General program synthesis.}
As explained in \cref{Sec:Intro}, general synthesis engines
(e.g., Sketch~\cite{Solarlezama2013}, Rosette~\cite{TorlakB14},
\sygus~\cite{AlurBodik2015}, and \semgus~\cite{KimHu2021}) are fundamentally different from \ipsmp.
Among these frameworks, only \semgus can both solve synthesis problems and  
prove unrealizability.
Similar to \ipsmp, \semgus reduces the synthesis problem to satisfiability of CHCs.
However, this is  where the similarities end. 
\semgus reduces synthesis to unsatisfiability and extracts solutions from the refutation proofs.
In contrast, \ipsmp reduces to satisfiability and solutions are extracted from model of the CHCs.
\semgus solves a more general problem, which comes at a high price
both from a theoretical and practical perspective.
We show that \ipsmp modulo Boolean programs can be solved in polynomial time (in the number of states), while \semgus lacks this guarantee.
Existing \semgus solvers (e.g., \textsc{Messy}~\cite{KimHu2021}) synthesize programs from sets of candidates described using regular tree grammars.
As a result, their CHCs use constraints over Algebraic Data Types to represent the grammar terms, which are harder to solve than either Boolean or linear arithmetic constraints.
Only Sketch and Rosette are ``modulo programs'', but do not allow loops nor recursion.

\noindent
\textbf{Specification synthesis.}
Specification synthesis solves the problem of finding specifications for unknown procedures which enable the verification of a given program (e.g.,~\cite{DasLahiri2015,AlbarghouthiDillig2016,PrabhuFedyukovich2021}).
Unlike \ipsmp, specification synthesis is under-specified.
Trivial specifications such as \emph{false} are often sufficient but undesirable.
As a consequence, many tools aim to synthesize either \emph{weakest} (i.e.,
maximal) or non-vacuous solutions.
In \ipsmp, any solution is valid as long as it satisfies all program assertions.
In \cref{Sect:Implementation}, we also compare our \ipsmp solver with
\hornspec~\cite{PrabhuFedyukovich2021} and demonstrate that \hornspec
is unsuitable for \ipsmp.

\noindent
\textbf{Data-driven invariant generation.}
Multiple approaches have been proposed (e.g.,~\cite{GargLoding2014,ZhuMagill2018,SiNaik2020,YaoRyan2020,RyanWong2020,HuangQiu2020}) that rephrase loop invariant synthesis as a learning problem.
Recent work has extended these techniques to parameterized verification~\cite{YaoTao2021}.
Often, these techniques require problem-specific biases to learn useful invariants~(e.g.,~\cite{SiNaik2020,YaoRyan2020,RyanWong2020,YaoTao2021}).
Furthermore, these techniques lack the complexity bounds of decidable verification.
In contrast, \ipsmp is problem-agnostic, and achieves the same complexity as verification in the Boolean case.
Adapting data-driven techniques to \ipsmp-solving is an interesting future direction.

\noindent
\textbf{Constrained Horn clauses.} In recent years, CHC-solvers have become a common tool for verification and synthesis problems.
Example include \seahorn~\cite{GurfinkelKahsai2015}, \semgus~\cite{KimHu2021}, and \hornspec~\cite{PrabhuFedyukovich2021}.
The connection between CHCs and verification has long been explored in the CLP community~(e.g., \cite{JaffarSV04,PeraltaGS98,DelzannoP99}).
This direction was popularized again by the work of Rybalchenko et al.~\cite{GrebenshchikovLPR12}.
According to the annual CHC-COMP competition\footnote{\url{https://chc-comp.github.io}.}, \spacer~\cite{KomuravelliGurfinkel2014} and \eldarica~\cite{HojjatRummer2018} are the most effective general-purpose CHC-solvers. 

%% file: conclusion.tex
\section{Conclusion}

We proposed \ipsmp, a novel synthesis problem suitable for solving a wide range of verification problems, such as invariant inference and verification of parameterized systems.
To demonstrate the relevance of \ipsmp, we provided three reductions from classic verification problems to \ipsmp.
To highlight \ipsmp's practicality, we proposed a solution that effectively leverages off-the-shelf CHC solvers and implemented it in the \seahorn verification framework.
Our evaluation demonstrates the effectiveness of CHC solvers in solving \ipsmp when compared with general synthesis tools such as \hornspec and \cvc.

Finally, we demonstrated that the interesting instance of \ipsmp for Boolean programs is efficiently decidable, whereas the general instance is undecidable.
Despite this, the general instance of \ipsmp is theoretically simpler than general synthesis, and thus, warrants specialized solvers.
In future work, we plan to study other instances of \ipsmp, such as IPS modulo timed automata.
We further suspect that \ipsmp will enable new practical applications of PCMC.

%% file: appendix_lang.tex
\section{Syntax and Semantics}
\label{Appendix:Language}

\begin{figure}[t]
  \input{diagrams/grammar}
  \caption{The formal grammar for programs with variables $\cV$ and operations over the signature $\Sigma$. That is, $\langle \textrm{Var} \rangle ::= \cV$, $\langle \textrm{Val} \rangle ::= \Terms( \Sigma, \cV )$, and $\langle \textrm{Expr} \rangle ::= \QFFml( \Sigma, \cV )$. The set of programs in the language is denoted by $\Programs( \Sigma, \cV )$.}
  \label{Fig:Appendix:Lang}
\end{figure}

\begin{figure}[t]
  \input{diagrams/wlp}
  \vspace{-0.25in}
  \caption{The WLP transformer for \cref{Fig:Appendix:Lang}. This follows the presentation of~\cite{BjornerGurfinkel2015}.}
  \label{Fig:Appendix:WLP}
 \end{figure}

The syntax of $\Programs( \Sigma, \cV )$ is presented in \cref{Fig:Appendix:Lang}.
To simplify the presentation, types are omitted and all local variables are declared as inputs to procedures.
Up to these simplifications, all IPS-MP instances in \cref{Sect:Overview} can be thought of as programs in this language.
A qualitative description of this language can be found in \cref{Sect:Background:Lang}.
Denotational semantics for this language are given by the WLP transformer in \cref{Fig:Appendix:WLP}.
In this semantic interpretation, each loop $S$ at line $\textit{ln}$ is associated with a predicate $\linv_{ln}$.
The result of $\wlp( S, Q )$ encodes that $\linv_{ln}$ is a loop invariant for $S$ which is sufficient to entail $Q$.

%% file: diagrams/grammar.tex
\begin{bnf*}
    \bnfprod{Procs}
            {\bnftd{a procedure name}
             \bnfor
             \bnfts{Main}}\\
    \bnfprod{Preds}
            {\bnftd{a predicate name}}\\
    \bnfprod{VarList}
            {\bnfpn{Var} \bnfts{,} \bnfsp \bnfsk \bnfts{,} \bnfsp \bnfpn{Var}}\\
    \bnfprod{ValList}
            {\bnfpn{Expr} \bnfts{,} \bnfsp \bnfsk \bnfts{,} \bnfsp \bnfpn{Expr}}\\
    \bnfprod{ProcApp}
            {\bnfpn{Procs} \bnfts{(} \bnfsp \bnfpn{ValList} \bnfsp \bnfts{)}}\\
    \bnfprod{PredApp}
            {\bnfpn{Preds} \bnfts{(} \bnfsp \bnfpn{ValList} \bnfsp \bnfts{)}}\\
    \bnfprod{Inst}
            {\bnfpn{Var} \bnfsp \bnfts{=} \bnfsp \bnfpn{Val}
             \bnfor
             \bnfpn{Var} \bnfsp \bnfts{=} \bnfsp \bnfts{*}
             \bnfor
             \bnfts{skip}
             \bnfor
             \bnfts{assume} \bnfts{(} \bnfsp \bnfpn{Expr} \bnfsp \bnfts{)}
             \bnfor}\\
    \bnfmore{
             \bnfts{assert} \bnfts{(} \bnfsp \bnfpn{Expr} \bnfsp \bnfts{)}
             \bnfor
             \bnfts{assume} \bnfts{(} \bnfsp \bnfpn{PredApp} \bnfsp \bnfts{)}
             \bnfor}\\
    \bnfmore{
             \bnfts{assert} \bnfts{(} \bnfsp \bnfpn{PredApp} \bnfsp \bnfts{)}
             \bnfor
             \bnfpn{VarList} \bnfsp \bnfts{=} \bnfsp \bnfpn{ProcApp}}\\
    \bnfprod{Stmt}
            {\bnfpn{Stmt} \bnfts{;} \bnfsp \bnfpn{Stmt}
             \bnfor
             \bnfpn{Inst}
             \bnfor
             \bnfts{while} \bnfsp \bnfts{(} \bnfsp \bnfpn{Expr} \bnfsp
                 \bnfts{)} \bnfsp \bnfts{\{} \bnfsp \bnfpn{Stmt} \bnfts{;} \bnfsp
                 \bnfts{\}}
             \bnfor}\\
    \bnfmore{
             \bnfts{if} \bnfsp \bnfts{(} \bnfsp \bnfpn{Expr} \bnfsp \bnfts{)}
                 \bnfsp \bnfts{\{} \bnfsp \bnfpn{Stmt} \bnfts{;} \bnfsp \bnfts{\}}} \bnfsp
                 \bnfts{else} \bnfsp \bnfts{\{} \bnfsp \bnfpn{Stmt} \bnfts{;} \bnfsp
                 \bnfts{\}}\\
    \bnfprod{ProcDecl}
            {\bnfpn{Procs} \bnfts{(} \bnfsp \bnfpn{VarList} \bnfsp \bnfts{)}
                 \bnfsp \bnfts{\{} \bnfsp \bnfpn{Stmt} \bnfts{;} \bnfsp
                 \bnfts{return} \bnfsp \bnfpn{ValList} \bnfts{;} \bnfsp
                 \bnfts{\}}}\\
    \bnfprod{PredDecl}
            {\bnfpn{Preds} \bnfts{(} \bnfsp \bnfpn{VarList} \bnfsp \bnfts{)}
                 \bnfsp \bnfts{\{} \bnfsp \bnfts{return} \bnfsp \bnfpn{Expr} \bnfts{;}
                 \bnfsp \bnfts{\}}}
\end{bnf*}

%% file: diagrams/wlp.tex
\renewcommand{\gets}{\mathbin{:=}}
\begin{gather*}
  \begin{aligned}
  \wlp( S_1; S_2, Q ) &\gets \wlp( S_1, \wlp( S_2, Q ) ) \\
  \wlp( \wlpwhile_{ln}\; ( \varphi )\; \{ S \}, Q ) &\gets \begin{aligned}[t]
    & \forall \vec{w} \cdot \left( \left( \linv_{ln}( \vec{w} ) \land \varphi \right) \implies \wlp\left( S, \linv_{ln}( \vec{w} ) \right) \right) \land \\
    & \forall \vec{w} \cdot \left( \left( \linv_{ln}( \vec{w} ) \land \neg \varphi \right) \implies Q \right) \land \linv_{ln}( \vec{w} ) \\
  \end{aligned} \\
  \wlp( \wlpif\; ( \varphi )\; \{ S_1 \}\; \wlpelse\; \{ S_2 \}, Q) &\gets
    \left( \varphi \implies \wlp( S_1, Q ) \right) \land
    \left( \neg \varphi \implies \wlp( S_2, Q ) \right) \\
  \wlp( x = *, Q ) &\gets \forall x \cdot Q \\
  \wlp( \vec{y} = f( \vec{e} ), Q ) &\gets
    f_{pre}( \vec{e} ) \land \forall \vec{r} \cdot \left(
      f_{sum}( \vec{e}, \vec{r} ) \implies Q[ \vec{y} / \vec{r} ]
    \right) \\
    \wlp( \wlpskip, Q ) &\gets Q \\
    \wlp( x = e, Q ) &\gets Q[ x / e ] \\
    \wlp( \wlpassert( \varphi ), Q ) &\gets \varphi \land Q \\
    \wlp( \wlpassume( \varphi ), Q ) &\gets \varphi \implies Q
  \end{aligned}
\end{gather*}

%% file: appendix_pcmc.tex
\section{Loop Invariant Inference as Synthesis}
\label{Appendix:Loop}

\begin{figure}[t]
  \begin{subfigure}{0.42\textwidth}
\begin{lstlisting}[style=ipsmp]
void main(int x) {
  assume(x >= 0);
  int i = 0; int y = 0;
  while (i < x) { i += 1; y += 2;} ¤\label{Line:Overview:LoopProb:Loop}¤
  assert(y == 2 * x); } ¤\label{Line:Overview:LoopProb:Return}¤
\end{lstlisting}
    \vspace{-0.1in}
    \caption{The original program.}
    \label{Fig:Overview:LoopProb}
  \end{subfigure}
  \hfill
  \begin{subfigure}{0.51\textwidth}
\begin{lstlisting}[style=ipsmp]
°bool PRED_TEMPLATE Inv1(int x, int i) {° ¤\label{Line:Overview:LoopSoln:Inv1}¤
  °return synth(x, i); }°
°bool PRED_TEMPLATE Inv2(int i, int y) {° ¤\label{Line:Overview:LoopSoln:Inv2}¤
  °return synth(i, y); }°
void main(int x) {
  assume(x >= 0);
  int i = 0; int y = 0;
  ~assert(Inv1(x, i)); assert(Inv2(i, y));~ ¤\label{Line:Overview:LoopSoln:EntryStart}¤¤\label{Line:Overview:LoopSoln:EntryEnd}¤
  ^y = *; i = *;^ ¤\label{Line:Overview:LoopSoln:SelectStart}¤
  ^assume(Inv1(x, i)); assume(Inv2(i, y));^ ¤\label{Line:Overview:LoopSoln:AssumeStart}¤¤\label{Line:Overview:LoopSoln:SelectEnd}¤
  if (i < x) { // Unrolled while loop. ¤\label{Line:Overview:LoopSoln:UnrollIn}¤
    i += 1; y += 2; ¤\label{Line:Overview:LoopSoln:ClosedStart}¤
    ~assert(Inv1(x, i)); assert(Inv2(i, y));~
    assume(false); } ¤\label{Line:Overview:LoopSoln:ClosedEnd}¤¤\label{Line:Overview:LoopSoln:UnrollOut}¤
  assert(y == 2 * x); } ¤\label{Line:Overview:LoopSoln:Exit}¤
\end{lstlisting}
    \vspace{-0.1in}
    \caption{The \ipsmp problem.}
    \label{Fig:Overview:LoopSoln}
  \end{subfigure}
  \caption{A program that is correct relative to the loop invariant $( 2 \cdot i = y ) \land (i \le x )$, and an \ipsmp problem that synthesizes the loop invariant.}
\end{figure}

Consider the program in \cref{Fig:Overview:LoopProb}.
This program takes as input a non-negative integer \code{x}, and then computes $2 \cdot x$ through repeated addition.
The function is correct if $y = 2 \cdot x$ on \auxlineref{Line:Overview:LoopProb:Return}.
A standard approach to this problem is to find an invariant for the loop on \auxlineref{Line:Overview:LoopProb:Loop} that entails $y = 2 \cdot x$ on \auxlineref{Line:Overview:LoopProb:Return}.
Therefore, the goal of this example is to construct an \ipsmp problem to synthesize such a loop invariant (existence of this loop invariant entails the correctness of \cref{Fig:Overview:LoopProb}).

By definition, a loop invariant is a predicate that is true upon entry to the loop, closed under each iteration of the loop, and true of the program's state upon loop termination~\cite{Gries1982}.
Each requirement of a loop invariant can be represented through assumptions and assertions, as in \cref{Fig:Overview:LoopSoln}.
First, to ensure that the loop invariant is true upon entry, the loop invariant is asserted upon entering the loop (\auxlineref{Line:Overview:LoopSoln:EntryStart}).
Second, to ensure that the loop invariant is closed under each iteration of the loop, \cref{Fig:Overview:LoopSoln} over-approximates the state of the program upon entry to an arbitrary iteration of the loop.
To restrict the program to an arbitrary iteration of the loop, the loop is first unrolled to a single iteration (\linerange{Line:Overview:LoopSoln:UnrollIn}{Line:Overview:LoopSoln:UnrollOut}).
Before checking the loop condition, the state of an arbitrary loop iteration is then selected by setting each mutable variable non-deterministically, and assuming that these new values satisfy the loop invariant (\linerange{Line:Overview:LoopSoln:SelectStart}{Line:Overview:LoopSoln:SelectEnd}).
If this state also satisfies the loop condition, then closure is enforced by first executing the body of the loop, and then asserting that the loop invariant is maintained (\linerange{Line:Overview:LoopSoln:ClosedStart}{Line:Overview:LoopSoln:ClosedEnd}).
Otherwise, the state does not satisfy the loop condition, and the program exits the loop (\auxlineref{Line:Overview:LoopSoln:Exit}).
This gives a program with unknowns, as required by the verification methodology.

Next, the shape of the loop invariant is restricted.
We note that the loop is correct because \code{i} is incremented up to \code{x}, while maintaining that \code{y} is twice the value of \code{i}.
It follows that there exists a loop invariant for \cref{Fig:Overview:LoopProb}
that does not relate \code{x} and \code{y} in the same constraint.
However, \code{x} must be compared with \code{i}, and \code{i} must be compared with \code{y}.
Therefore, our loop invariant has the shape $\code{Inv1}( x, i ) \land \code{Inv2}( i, y )$.
In the \ipsmp encoding, both \code{Inv1} and \code{Inv2} correspond to partial predicates (see lines~\ref{Line:Overview:LoopSoln:Inv1} and \ref{Line:Overview:LoopSoln:Inv2}, respectively) that  are assumed and asserted together (see lines~\ref{Line:Overview:LoopSoln:AssumeStart} and \ref{Line:Overview:LoopSoln:EntryStart}, respectively).
One solution to this problem is to assign the expression \code{(i <= x)} to the hole in \code{Inv1}, and the expression \mbox{\code{(y == 2 * i)}} to the hole in \code{Inv2}.

\section{Loop Invariant Inference: Reduction to \ipsmp}
\label{Appendix:LoopReduction}

A \emph{safe loop invariant} is a predicate that is true upon entry to a loop, maintained by each iteration of a loop, and is sufficient to prove that the program is correct~\cite{Gries1982}.
\emph{Loop invariant inference} asks to find a safe loop invariant given a program.
The inference problem is \emph{intensional} if solutions must be in the same logical fragment as assertions in the programming language~\cite{NielsonNielson2007}.
A formal definition of (intensional) safe loop invariant inference is given in \cref{Def:LoopInvInference}.
Note that in \cref{Def:LoopInvInference}, $\ToCHC( f )$ is used to relate the loop invariant $\varphi$ to the summary of each procedure $f$ in $\cP$.
For simplicity of presentation, a program has a single loop and two variables (\cref{Fig:Reduction:LoopProb}).
A generalization to $m$ loops is possible and not difficult.
A generalization to $n$ variables follows immediately.

\begin{definition}
  \label{Def:LoopInvInference}
  A \emph{loop invariant inference problem} is a tuple $( \cP, \cT )$ such that $\cP \in \Programs( \Sigma, \{ x_1, x_2 \} )$ is a problem with the \code{main} procedure from \cref{Fig:Reduction:LoopProb} (where $\Sigma$ is a first-order signature), and $\cT$ is a theory.
  A solution to $( \cP, \cT )$ is a $\varphi \in \QFFml( \Sigma, \{ x_1, x_2 \} )$ such that the following are $\cT$-satisfiable:
  \begin{enumerate}
  \item $\psi_{\mathit{Pre}} := \forall \,V \cdot \wlp( S_1, \varphi )$, where $S_1$ is the statement before the loop;
  \item $\psi_{\mathit{Tr}} := \forall \,V \cdot ( \varphi \land e ) \implies \wlp( S_2, \varphi )$, where $S_2$ is the loop body and $e$ is the loop condition;
  \item $\psi_{\mathit{Post}} := \forall \,V \cdot ( \varphi \land \neg e ) \implies \wlp( S_3, \top )$, where $S_3$ is the statement after the loop;
  \item $\psi_{\mathit{Procs}} := \bigwedge_{f \in \Procs( \cP )} \ToCHC( f )$.
  \end{enumerate}
\end{definition}

\begin{theorem}
  \label{Thm:LoopReduction}
  Let $( \cP, \cT )$ be a loop invariant inference problem and $\cP'$ be the program obtained by replacing \code{main} in $\cP$ with the definition of \code{main} in \cref{Fig:Reduction:LoopSoln}.
  The partial predicate implementation $\Pi$ is a solution to $( \cP', \Pi_{\bot}, \cT )$ if and only if $\Pi( \code{Inv} )$ is a solution to $( \cP, \cT )$.
\end{theorem}

\begin{proof}
  Let $\Pi$ be a solution $( \cP', \Pi_{\bot}, \cT )$.
  This is true if and only if $\cP'[\Pi]$ is correct relative to $\cT$.
  By \cref{Prop:Safety}, this is true if and only if $\ToCHC( \cP'[\Pi] )$ is $\cT$-satisfiable.
  By definition:
  {\par\footnotesize\begin{align*}
    \ToCHC( \cP'[\Pi] ) &= \wlp( \cP'( \mmsmallcode{main}  ), \top ) \land \psi_{\mathit{Procs}} \land \forall V \cdot \left( \mmsmallcode{Inv}( x, y ) \iff \varphi \right) \\
    \wlp( \cP'( \mmsmallcode{main} ), \top ) &= \begin{aligned}[t]
        & \forall V \cdot \wlp( S_1, \mmsmallcode{Inv}( x, y ) \land \forall V' \cdot (
        \mmsmallcode{Inv}( x', y' ) \implies \\
        & ( ( e' \implies \wlp( S_2, \mmsmallcode{Inv}( x', y' ) ) ) \land ( \neg e' \implies \wlp( S_3, \top ) ))))
    \end{aligned}
  \end{align*}}%
  As $\forall \,V \cdot \code{Inv}( x, y ) \iff \varphi$, then $\ToCHC( \cP'[\Pi] )$ is $\cT$-satisfiable if and only if $\ToCHC( \cP'[\Pi] )$ is $\cT$-satisfiable after substituting $\varphi$ for \code{Inv}.
  By substituting $\varphi$ for \code{Inv} and rewriting $\wlp( \cP'( \code{main} ), \top )$ as a conjunction of CHCs, it follows that $\ToCHC( \cP'[\Pi] ) \iff \varphi_{\mathit{Pre}} \land \varphi_{\mathit{Tr}} \land \varphi_{\mathit{Post}} \land \varphi_{\mathit{Procs}}$.
  Therefore, $\varphi$ is a solution to $( \cP, \cT )$.
  The other direction is symmetric.
\end{proof}

\begin{figure}[t]
  \begin{subfigure}{0.47\textwidth}
\begin{lstlisting}[style=ipsmp]
void main(int x1, int x2) {
  // Code before loop.
  S1;
  // Loop condition and body.
  while (e) { S2; }
  // Code after loop.
  S3;
}
\end{lstlisting}
    \vspace{-0.1in}
    \caption{The input program.}
    \label{Fig:Reduction:LoopProb}
  \end{subfigure}
  \hfill
  \begin{subfigure}{0.47\textwidth}
\begin{lstlisting}[style=ipsmp]
°bool PRED_TEMPLATE Inv(int x1, int x2) {°
  °return synth(x1, x2); }°
void main(int x1, int x2) {
  S1;
  ~assert(Inv(x1, x2));~
  ^x1 = *; x2 = *;^
  ^assume(Inv(x1, x2));^
  if (e) {
    S2;
    ~assert(Inv(x1, x2));~
    assume(false);
  }
  S3;
}
\end{lstlisting}
    \vspace{-0.1in}
    \caption{The \ipsmp reduction.}
    \label{Fig:Reduction:LoopSoln}
  \end{subfigure}
  \caption{A reduction from loop invariant inference to \ipsmp. In both programs, $( S_1, S_2, S_3 )$ are statements, and $e$ is a Boolean expression in the language grammar.}
\end{figure}

%% file: appendix_proofs.tex
\section{Proof of Theorem \ref{Thm:BoolComplexity}} \label{Appendix:BoolComplexityProof}

\begin{proof}
  Let $k = |\GV \cup \LV|$.
  The call to \code{Init} on \auxlineref{Line:BoolSynth:InitAll} of \cref{Alg:BoolSynth} has complexity $O( |\Locs| + |\Bsum| )$ since \code{Init} iterates over $\Partial( \cP )$ and calls \code{InitBoolReach}, \code{InitBoolReach} iterates over $Locs$ and $\Bsum$, and $|\Locs| \ge \Partial( \cP )$.
  The loop on \auxlineref{Line:BoolSynth:Loop} of \cref{Alg:BoolSynth} requires $O \left( |\Locs| \cdot 2^{2 k} \right)$ iterations, since each state is visited at most once, there are at most $2^{2 k}$ variable assignments, and exactly $|\Locs|$ control-flow locations for each assignment.
  During each iteration, six procedures are called with the following worst-case runtime complexities:
  \begin{itemize}
  \item Each call to \code{DoIntraproc} performs $O( |\Bops| )$ operations since \code{DoIntraproc} iterates over $\Bops$ and does $O(1)$ operations per edge in $\Bops$.
  \item Each call to \code{DoProcs} performs $O( |\Bcalls| )$ operations since \code{DoProcs} iterates over $\Bcalls$ and does $O(1)$ operations per edge in $\Bcalls$.
  \item Each call to \code{DoAssumes} performs $O( |\Bassume| )$ operations since \code{DoAssumes} iterates over $\Bassume$ and does $O(1)$ operations per edge in $\Bassume$.
  \item Each call to \code{DoAsserts} performs $O( |\Bassert| )$ operations since \code{DoAsserts} iterates over $\Bassert$ and does $O(1)$ operations per edge in $\Bassert$.
  \item Each call to \code{DoProcSum} iterates over $\Bsum$.
        However, \code{DoProcSum} only processes $( \lin, \lout ) \in \Bsum$ if $\lout = \lwork$.
        Since each function has a single exit location, there is at most one edge in $\Bsum$ such that $\lout = \lwork$.
        Therefore, each call to \code{DoProcSum} performs $O( |\Bcalls| )$ operations, as once $( \lin, \lout )$ is found, \code{DoProcSum} does $O(1)$ operations, iterates over $\Bcalls$, and does $O(1)$ operations per edge in $\Bcalls$.
  \item Each call to \code{DoFuncSum} performs $O( |\Bassume \cup \Bassert| )$ operations since it iterates over both $\Bassume$ and $\Bassert$, and does $O(1)$ operations per edge in $\Bassume \cup \Bassert$.
  \end{itemize}
  Therefore, \code{BoolSynth} terminates within $O \left( 2^{2 k} \cdot |\Locs| \cdot |\Bops \cup \Bcalls \cup \Bassume \cup \Bassert| \right)$ operations.
\end{proof}

\section{Proof of Theorem \ref{Thm:Analyze}}
\label{Appendix:AnalyzeProof}

\begin{proof}
  From the correctness of \code{ComputeBoolReach} (\cref{Alg:BebopReach}):
  \begin{enumerate}
  \item \code{DoInterproc} encodes rule~3 of \cref{Def:Summary};
  \item \code{DoProcs} encodes rules~4~and~5 of \cref{Def:Summary};
  \item \code{DoProcSum} encodes rule~6 and maintains rule~5 of \cref{Def:Summary}.
  \end{enumerate}
  \cref{Alg:BoolSynth} introduces three new procedures, that correspond to the rules of \cref{Def:PartialSummary}.
  \begin{enumerate}
  \item \code{DoAssumes} encodes rule~5 of \cref{Def:PartialSummary};
  \item \code{DoAsserts} encodes rule~3, and queues work for \code{DoFunSum};
  \item \code{DoFunSum} maintains rules~4~and~5 of \cref{Def:PartialSummary} for items queued in \code{DoAsserts}.
  \end{enumerate}
  Following the proof of \cref{Alg:BebopReach}, the loop on \auxlineref{Line:BoolSynth:Loop} computes a fixed point of the equations in \cref{Def:PartialSummary}.
  This is the least solution that is weaker than the initial assignment to $( \theta, \sigma, \Pi )$ on \auxlineref{Line:BoolSynth:InitAll}.
  Both $\theta$ and $\sigma$ are initialized according to \code{ComputeBoolReach}, and follow rule~2 of \cref{Def:Summary}.
  The initial assignment to $\Pi$ is $\Pi_0$.
  Therefore, when \code{Analyze} terminates, $( \theta, \sigma, \Pi )$ is a least partial program summary such that $\forall p \in \Partial( \cP ) \cdot \Pi( p ) \implies \Pi'( p )$.
\end{proof}

\section{Proof of Corollary \ref{Cor:BoolCorrectness}}
\label{Appendix:BoolCorrectnessProof}

The following proposition was stated informally in \cref{Sect:Background:WLP}.
\begin{proposition}[\cite{AlurBouajjani2018}]
  \label{Prop:BoolSafety}
  If $( \theta, \sigma )$ is the least summary of a Boolean program $\cP$, then $\cP$ is correct if and only if $\sigma( \lbot ) = \bot$.
\end{proposition}
The proof of \cref{Cor:BoolCorrectness} follows.

\begin{proof}
  Assume that $( \cP, \varnothing, \Pi_0 )$ is an \ipsmp problem for a Boolean program $\cP$.
  By \cref{Thm:Analyze}, \auxlineref{Line:BoolSynth:CallAnalyze} computes a least partial program summary $( \theta, \sigma, \Pi )$ for $\cP$ such that $\forall p \in \Partial( \cP ) \cdot \Pi_0( p ) \implies \Pi( p )$.
  Furthermore, by \cref{Thm:BoolComplexity}, the call on \auxlineref{Line:BoolSynth:CallAnalyze} always terminates.
  Then, $\Pi$ solves the \ipsmp instance if and only if $\cP$ is safe.
  By \cref{Prop:BoolSafety}, $\cP[\Pi]$ is safe if and only if $\theta( \lbot ) = \bot$.
  On \auxlineref{Line:BoolSynth:Decide}, if $\theta( \lbot ) = \bot$, then a solution is returned, else a witness to unrealizability is returned.
  Therefore, \code{BoolSynth} decides the \ipsmp problem for Boolean programs.
\end{proof}

\section{Proof of Theorem \ref{Thm:Realizability}}
\label{Appendix:RealizabilityProof}

The following proposition was stated informally in \cref{Sect:Background:WLP}.
\begin{proposition}[\cite{BjornerGurfinkel2015}]
    \label{Prop:Safety}
    A program $\cP$ is correct relative to theory $\cT$ if and only if $\ToCHC( \cP )$ has a $\cT$-model.
\end{proposition}
The proof of \cref{Thm:Realizability} follows.

\begin{proof}
  Let $\chcbound( \cP ) := \bigwedge_{p \in \Partial( \cP )} \forall \vec{x} \cdot \left( \Pi_0( p ) \implies p( \vec{x} ) \right)$.
  Assume that $\sigma$ is an $\cF$-solution to $\SynthToCHC( \cP )$.
  Since $\chcbound( \cP )$ is a term of $\SynthToCHC( \cP )$, then $\forall p \in \Partial( \cP ) \cdot \models_{\cT} \Pi_0( p ) \implies \Pi( p )$.
  It then follows by induction on the size of $\Partial( \cP )$ that $\cP[ \Pi ]$ is correct.
  \begin{itemize}
  \item \textbf{Hypothesis}.
        For some $k \ge 0$, if $|\Partial( \cP )| \le k$ and $\SynthToCHC( \cP, \Pi_0 )$ has an $\cF$-solution, then $\cP$ is correct.
  \item \textbf{Base Case}.
        If $|\Partial( \cP )| = 0$, then $\cP$ is correct by \cref{Prop:Safety}.
  \item \textbf{Inductive Case}.
        Assume that $|\Partial( \cP )| = k + 1$, $\SynthToCHC( \cP )$ has an $\cF$-solution $\sigma$, and the inductive hypothesis holds up to $k$.
        Let $p \in \Partial( \cP )$ and $\cP' = \cP[ p \leftarrow \sigma( p ) ]$.
        By the definition of an interpretation, $\sigma$ is also a solution to $\ToCHC( \cP ) \land \chcbound( \cP ) \land \forall \vec{x} \cdot p( \vec{x} ) \iff \sigma( p ).$
        Furthermore, $\forall \vec{x} \cdot \Pi_0( p ) \implies p( \vec{x} )$ is subsumed by $\forall \vec{x} \cdot p( \vec{x} ) \iff \sigma( p )$.
        Therefore, $\sigma$ is an $\cF$-solution to $\SynthToCHC( \cP' )$.
        By hypothesis, $\cP'[ \Pi ]$ is correct.
        Since $\Pi(p) = \sigma(p)$, then $\cP'[ \Pi ] = \cP[ p \leftarrow \sigma(p) ][ \Pi ] = \cP[ \Pi ]$.
        Therefore, $\cP[ \Pi ]$ is also correct.
  \end{itemize}
  Therefore, $\Pi$ is a solution to $( \cP, \cT, \Pi_0 )$.
\end{proof}

\section{Proof of Theorem \ref{Thm:Unrealizability}}
\label{Appendix:UnrealizabilityProof}

The following proposition was stated informally in \cref{Sect:Background:WLP}.
\begin{proposition}[\cite{BjornerGurfinkel2015}]
    \label{Prop:CHCSafety}
    $\ToCHC( \cP )$ is a CHC conjunction.
\end{proposition}
The proof of \cref{Thm:Unrealizability} follows.

\begin{proof}
  Assume for the intent of contradiction that $\Pi$ is a solution to $( \cP, \cT, \Pi_0 )$.
  Then $\cP[\Pi]$ is correct relative to $\cT$, since $\Pi$ is a solution.
  Then by \cref{Prop:CHCSafety}, $\ToCHC( \cP[\Pi] )$ is $\cT$-satisfiable, since $\cP[\Pi]$ is correct relative to $\cT$.
  By definition:
  {\footnotesize\begin{equation*}
    \ToCHC( \cP[\Pi] ) = \ToCHC( \cP ) \land \left( \bigwedge_{p \in \Partial( \cP )} \forall \vec{x} \left( \Pi( p ) \iff p( \vec{x} ) \right) \right)
  \end{equation*}}%
  Then, $\ToCHC( \cP ) \land \left( \bigwedge_{p \in \Partial( \cP )} \forall \vec{x} \left( \Pi_0( p ) \implies p( \vec{x} ) \right) \right)$ has a $\cT$-solution, since $\Pi$ is a solution to $( \cP, \cT, \Pi_0 )$ and therefore satisfies $\forall p \in \Partial( \cP ) \cdot \models_{\cT} \Pi( p ) \implies \Pi_0( p )$.
  Then $\SynthToCHC( \cP, \Pi_0 )$ is $\cT$-satisfiable.
  By contradiction, $( \cP, \cT, \Pi_0 )$ is unrealizable.
\end{proof}

\section{Proof of Theorem \ref{Thm:SynthCHC}}
\label{Appendix:SynthCHCProof}

\begin{proof}
    By \cref{Prop:CHCSafety}, $\ToCHC( \cP )$ is a CHC conjunction.
    For each $p \in \Partial( \cP )$, $\Pi_0( p ) \implies p( \vec{x} )$ is a CHC, since $\Pi_0( p )$ is quantifier-free and $p$ is a predicate symbol.
    Therefore, $\SynthToCHC( \cP, \Pi_0 )$ is a CHC conjunction.
\end{proof}

\section{Proof of Theorem \ref{Thm:Undecidable}}
\label{Appendix:UndecidableProof}

The following proposition was stated informally in \cref{Sect:Decidability:General}.
\begin{proposition}[\cite{Minsky1967}]
  \label{Prop:2CMHalting}
  The halting problem is undecidable for $2$-counters.
\end{proposition}
The proof of \cref{Thm:Undecidable} follows.

\begin{proof}
  Assume for the intent of contradiction that \ipsmp is decidable for linear integer arithmetic.
  Let $\cV = \{ x, y, z \}$, $\Sigma$ be the signature of linear integer arithmetic, and $\cT$ be the theory of linear integer arithmetic.
  Every $2$-counter machine can be encoded in $\Programs( \Sigma, \cV )$ as follows:
  \begin{enumerate}
  \item The two integer counters are $x$ and $y$.
  \item The program counter is $z$.
  \item The body of $\cP( \code{main} )$ is a \code{while} loop with loop condition \code{true}.
  \item The body of the \code{while} loop is a sequence of \code{if}-\code{else} statements that maps each value of $z$ to an instruction.
  \item The instruction \code{inc(x)} is: \code{x = x + 1; z = z + 1;}
  \item The instruction \code{dec(x)} is: \code{x = x - 1; z = z + 1;}
  \item The instruction \code{jump(x, i)} is: \code{if (x==0) \{ z=i; \} else \{ z=z+1; \}}
  \item The instruction \code{halt()} is: \code{assert(false);}
  \end{enumerate}
  If $\cP \in \Programs( \Sigma, \cV )$ is a $2$-counter machine (following the above encoding), then $\cP$ halts if and only if $\cP$ violates an assertion.
  Furthermore, $|\Partial( \cP )| = 0$.
  Let $\Pi_0$ be the trivial function from $\Partial( \cP )$ to $\QFFml( \Sigma, \cV )$.
  Then the \ipsmp problem $( \cP, \cT, \Pi_0 )$ has a solution if and only if $\cP$ halts.
  Since \ipsmp is decidable, then the halting problem for $2$-counter machines is decidable.
  However, the halting problem is undecidable for $2$-counter machines by \cref{Prop:2CMHalting}.
  Then by contradiction, \ipsmp is undecidable for the theory of integer linear arithmetic.
\end{proof}

\section{Proof of Theorem \ref{Thm:ClassReduction}}
\label{Appendix:ClassReductionProof}

\begin{proof}
  Let $\Pi$ be a solution $( \cP', \Pi_{\bot}, \cT )$.
  This is true if and only if $\cP'[\Pi]$ is correct relative to $\cT$.
  By \cref{Prop:Safety}, this is true if and only if $\ToCHC( \cP'[\Pi] )$ is $\cT$-satisfiable.
  By definition
  {\footnotesize\begin{align*}
    \ToCHC( \cP'[\Pi] ) &= \wlp( \cP'( \mmsmallcode{main} ), \top ) \land \ToCHC( \cP ) \land \forall V \cdot \left( \mmsmallcode{Inv}( x, y ) \iff \varphi \right) \\
    \wlp( \cP'( \mmsmallcode{main} ), \top ) &= \forall V \cdot ( ( ( \mmsmallcode{br} = 0 ) \implies \tau_0 ) \land \cdots \land ( ( \mmsmallcode{br} = 3 ) \implies \tau_3 ) )
  \end{align*}}%
  where $\tau_i$ is the WLP of the $i$-th branch of $\cP'( \code{main} )$.
  Since \code{br} does not appear in any $\tau_i$, then $\wlp( \cP'( \code{main} ), \top )$ is equisatisfiable with $( \forall V \cdot \tau_0 ) \land \cdots \land ( \forall V \cdot \tau_3 )$.
  Since $\forall V \cdot \left( \code{Inv}( x, y ) \iff \varphi \right)$, $\ToCHC( \cP'[\Pi] )$ is $\cT$-satisfiable if and only if $\ToCHC( \cP'[\Pi] )$ is $\cT$-satisfiable after substituting $\varphi$ for \code{Inv}.
  Observe that, after substituting $\varphi$ for \code{Inv} and simplifying to CHC conjunctions:
  {\footnotesize\begin{align*}
      \tau_0 &= \psi_{\mathit{Init}}
      &
      \tau_1 &= \psi_{\mathit{Closure1}} \land \forall V \cdot \left( \varphi \implies \mmsmallcode{f}_{pre}( x, y, a ) \right) \\
      \tau_2 &= \psi_{\mathit{Closure2}} \land \forall V \cdot \left( \varphi \implies \mmsmallcode{g}_{pre}( x, y, a, b ) \right)
      &
      \tau_3 &= \psi_{\mathit{Sufficient}} \land \forall V \cdot \left( \varphi \implies \mmsmallcode{func}_{pre}( x, y, a ) \right)
  \end{align*}}%
  As a direct result, $\ToCHC( \cP[\Pi'] ) \iff \psi_{\mathit{Init}} \land \psi_{\mathit{Closure1}} \land \psi_{\mathit{Closure2}} \land \psi_{\mathit{Sufficient}} \land \psi_{\mathit{Sum}}$.
  Therefore, $\varphi$ is a solution to $( \cP, \cT )$.
  The other direction is symmetric.
\end{proof}

\section{Proof of Theorem \ref{Thm:ParamReduction}}
\label{Appendix:PCMCReductionProof}

\begin{proof}
  Let $\Pi$ be a solution $( \cP', \Pi_0, \cT )$.
  This is true if and only if $\cP'[\Pi]$ is correct relative to $\cT$.
  By \cref{Prop:Safety}, this is true if and only if $\ToCHC( \cP'[\Pi] )$ is $\cT$-satisfiable.
  By definition
  {\footnotesize\begin{align*}
    \ToCHC( \cP'[\Pi] ) &= \wlp( \cP'( \mmsmallcode{main} ), \top ) \land \psi_{\mathit{Process}} \land \forall V \cdot \left( \mmsmallcode{Inv}( l, s, r ) \iff \varphi \right) \\
    \wlp( \cP'( \mmsmallcode{main} ), \top ) &= \forall V \cdot \left( \left( \left( \mmsmallcode{br} = 0 \right) \implies \tau_0 \right) \land \left( \left( \mmsmallcode{br} \ne 0 \right) \implies \tau_1 \right) \right),
  \end{align*}}%
  where $\tau_i$ is the WLP of the $i$-th branch of $\cP'( \code{main} )$.
  Since \code{br} does not appear in $\tau_0$ or $\tau_1$, then $\wlp( \cP'( \code{main} ), \top )$ is equisatisfiable with $( \forall V \cdot \tau_0 ) \land ( \forall V \cdot \tau_1 )$.
  By definition of $\wlp$:
  {\footnotesize\begin{align*}
    \tau_0 =&\, \left( \mmsmallcode{Inv}( l, s, r ) \land \mmsmallcode{Inv}( r, i, l ) \right) \implies \\
           &\, \left( \mmsmallcode{tr}_{pre}( l, s, r ) \land \left( \mmsmallcode{tr}_{sum}( l, s, r, l', s', r' ) \implies \left( \mmsmallcode{Inv}( l', s', r' ) \land \mmsmallcode{Inv}( r', i, l' ) \right) \right) \right)
  \end{align*}}%
  Since $\forall V \cdot \left( \code{Inv}( x, y ) \right) \iff \varphi$, $\ToCHC( \cP'[\Pi] )$ is $\cT$-satisfiable if and only if $\ToCHC( \cP'[\Pi] )$ is $\cT$-satisfiable after substituting $\varphi$ for \code{Inv}.
  Observe that, after substituting $\varphi$ for \code{Inv} and simplifying to CHC conjunctions:
  {\footnotesize\begin{align*}
    \tau_0 &= \psi_{\mathit{Closure}} \land \psi_{\mathit{Inf}} \land \forall V \cdot \left( \varphi \land \varphi_{\mathit{Inf}} \implies \mmsmallcode{tr}_{pre}( l, s, r) \right) &
    \tau_1 &= \psi_{\mathit{Adequate}}
  \end{align*}}%
  Then $\ToCHC( \cP[\Pi'] ) \iff \psi_{\mathit{Closure}} \land \psi_{\mathit{Inf}} \land \psi_{\mathit{Adequate}} \land \psi_{\mathit{Sum}}$.
  Since the template for \code{Inv} is \code{init}, then also $\models_{\cT} \code{init}( l, s, r ) \implies \varphi$.
  Therefore, $\varphi$ is a solution to $( \cP, \cT )$.
  The other direction is symmetric.
\end{proof}